\documentclass[a4paper,10pt,notitlepage]{article}
\usepackage[notcite,notref]{}
\usepackage{latexsym}
\usepackage{amssymb}
\usepackage{color}
\usepackage{graphicx}
\usepackage{amsmath}
\usepackage{amsthm}
\usepackage{mathrsfs}
\usepackage{natbib}

\usepackage[notcite,notref]{}

\newtheorem{theorem}{Theorem}[section]
\newtheorem{definition}[theorem]{Definition}
\newtheorem{lemma}[theorem]{Lemma}
\newtheorem{remark}[theorem]{Remark}
\newtheorem{assumption}[theorem]{Assumption}
\newtheorem{corollary}[theorem]{Corollary}

\def\P{\mathbb P}
\def\E{\mathbb E}
\def\R{\mathbb R}
\def\A{\mathcal X}
\def\At{\mathcal X(t,z,x)}
\def\Ata{\mathcal X(t,z,1)}

\def\QH{\mathbf{QH}}
\def\Maxmu{\overline \mu}
\def\Klip{K_{lip}^c}
\def\Klid{K_{lip}^d}
\def\Kmax{K_{max}}
\def\Klipa{K_{lip}^a}

\newcommand{\HZR}[2][]{H^{#1}(\mathbb{R}^{#2})}
\newcommand{\HZpar}[2][]{H^{#1}([0,T]\times \mathbb{R}^{#2})}

\def\semim{semimartingale }
\def\hold{H\"older }
\def\lip{Lipschitz }

\def\lev{L\'evy }
\def\itos{It\^o's }

\newcommand{\op}[2][]{\mathcal{#1}_t #2}
\newcommand{\genop}[2][]{\mathcal{#1} #2}

\def \q {
\begin{flushright}
 $\square$
\end{flushright}
}

\usepackage[a4paper,top=3cm,bottom=3cm,left=3cm,right=3cm]{geometry}

\title{Numerical methods for the quadratic hedging problem in Markov
  models with jumps\footnote{We thank Nadia Oudjane for illuminating discussions.}}
\author{Carmine De Franco\footnote{OSSIAM, E-mail:
    carmine.de-franco@ossiam.com}\and Peter Tankov\footnote{LPMA,
    Université Paris Diderot -- Paris 7, E-mail: tankov@math.univ-paris-diderot.fr}\and Xavier Warin\footnote{EDF R\&D, 92141 Clamart, France and Laboratoire de Finance des March\'es de l'Energie, Universit\'e Paris
Dauphine. Email: xavier.warin@edf.fr}}
\begin{document}
\frenchspacing

\maketitle
\begin{abstract}
 We develop algorithms for the numerical computation of the
quadratic hedging strategy in incomplete markets modeled by
pure jump Markov process. Using the Hamilton-Jacobi-Bellman approach,
the value function of the quadratic hedging problem can be related to
a triangular system of parabolic  partial integro-differential
equations (PIDE), which can be shown to possess unique smooth solutions in
our setting. The
first equation is non-linear, but does not depend on the pay-off of
the option to hedge (the pure investment problem), while the other two
equations are linear. We propose convergent finite difference schemes
for the numerical solution of these PIDEs and illustrate our
results with an application to electricity markets, where
time-inhomogeneous pure jump Markov processes appear in a natural
manner. 
\end{abstract}
\medskip

\medskip

\section{Introduction}\label{intro}
\numberwithin{equation}{section}

In an incomplete market setting, where exact replication of contingent
claims is not possible, quadratic hedging is the most common approach,
among both academics and
practitioners.  This method consists in
minimizing the $\mathbb L^2$ distance between the hedging portfolio and the
claim. Its popularity is due to the fact that the strategy is
linear with respect to the claim, and is relatively easy to compute in
a variety of settings. 

In its most general form, the quadratic hedging problem can be
formulated as follows. Consider a random variable $H\in\mathbb
L^2\left(\mathcal F_T,\P\right)$ (which stands for the option one
wants to hedge) and a set $\A$ of admissible strategies, which is a
subset of the set of adapted processes with caglad paths. The quadratic hedging problem becomes
\begin{align}
\label{quadratic_error}
\text{minimize   } \E^{\P}\left[\left(x+\int_0^T\theta_t dS_t -H \right)^2\right]\, \text{over $x\in\R$ and $\theta\in\A$}
\end{align}
where $S$ is a \semim modeling the stock price. If $(x^*,\theta^*)$
is a minimizer, we call $\theta^*$ the optimal mean-variance
hedging strategy and $x^*$ its price. This problem has been
extensively studied in the literature, starting with the seminal works
of \cite{follmer.sondermann.86} and \cite{follmer.schweizer.91} and
until the complete theoretical solution in the general
semimartingale setting was given in \cite{cern.kall.07}. The case when
$S$ is a $\P$-square integrable martingale is particularly simple and
can be solved using the well known Galtchouk-Kunita-Watanabe
decomposition. The general case is much more involved, and has only
been solved in \cite{cern.kall.07} by means of introducing a specific
non martingale change of measure (the opportunity neutral measure).

The problem of  numerical computation of the hedging strategy is an
important issue in its own right, since various objects appearing in
the theoretical solution (opportunity neutral measure, Galtchouk-Kunita-Watanabe
decomposition, F\"ollmer-Schweizer decomposition) are often not known
in explicit form. When the underlying asset is
modeled by a L\'evy process, a complete semi-explicit solution was
obtained in \cite{huba.kall.kraw.06} using Fourier
methods. Their approach was extended to additive processes in \cite{goutte.oudjane.russo.11}. \cite{laur.pham.99} and
\cite{heath.al.01} characterize the optimal strategy via
an HJB equation in continuous Markovian stochastic volatility models while
\cite{cern_kall_08} and \cite{kall.vier.09} treat affine stochastic
volatility models using Fourier methods. 

In this paper, we propose algorithms for the numerical computation of
the quadratic hedging strategy in general Markovian models with
jumps. We first review the HJB characterization of the value function,
obtained in \cite{defr.12}. We only give a brief
review, referring the readers to \cite{defr.12} for full details and
proofs because in this paper, we are interested in the numerical schemes for the
computation of the hedging strategies and in applications to
electricity markets. The value function of the quadratic hedging problem can be related to
a triangular system of parabolic PIDEs, which can be shown to
possess unique smooth solutions in our setting. The
first equation is a non-linear PIDE of HJB type, but does not depend on the pay-off of
the option to hedge (the pure investment problem), while the other two
equations are linear. We next propose two finite difference schemes
for the numerical solution of the linear and the nonlinear PIDEs. The
convergence of these schemes is carefully analyzed and we provide an
estimate of the global approximation error as function of various
truncation and discretization parameters. For the numerical
schemes, we concentrate on the infinite variation case, which is more
relevant in applications. 

Our main motivation comes from hedging problems in electricity
markets. These markets are structurally incomplete and often illiquid,
owing to a relatively small number of market participants and the
particular nature of electricity, which is a non-storable
commodity. As pointed out in \cite{geman.roncoroni.06} and \cite{mbt.08}, due to these features, the electricity prices exhibit
highly non-Gaussian behavior with jumps and spikes (upward movements
followed by quick return to the initial level)
and several authors have therefore suggested
to model electricity prices by pure jump processes \citep{benth.al.07,deng.jiang.05}.

On the other hand, since the spot electricity is non storable, the
main hedging instruments in electricity markets are futures.  
A typical future contract with maturity $T$ and duration $d$
guarantees to its holder continuous delivery of electricity during
the period $[T,T+d]$. Maturities, durations and amounts of electricity
are standardized for listed contracts. This continuous delivery
feature implies that even if the spot electricity follows a simple
model, such as the exponential of a L\'evy process, the price of the
future contract will be a general Markov process with jumps,
non-homogeneous in time and space.\footnote{If a single delivery length is
  fixed, it is possible to model the future price directly as the
  exponential of a process with independent increments, as in
  \cite{goutte.oudjane.russo.11}. However this approach does not allow
to treat problems involving futures of different durations, say,
hedging a product with specified duration with the listed contracts.} Therefore, Fourier methods such as
the ones developed in \cite{huba.kall.kraw.06} and \cite{goutte.oudjane.russo.11} cannot be applied in
this setting. For this reason, in Section \ref{sec: appl}, after
introducing a model for the futures prices,  where the spot price is
described by the exponential of the tempered stable (CGMY) or Normal Inverse Gaussian (NIG)
L\'evy process, we derive the associated HJB equations and use the finite
difference schemes to compute the hedging strategies and analyze their
behavior. The numerical results illustrate the performance of our
method and show in particular that the computation of the hedging
strategies under the true historical probability (as opposed to the
martingale probability, which does not require solving non-linear HJB
equations) leads to a considerable improvement in the efficiency of
the hedge.

The paper is structured as follows. After introducing the model and the
quadratic hedging problem in Section \ref{model}, we 
review the HJB characterization of the solution and the
regularity results in Section \ref{case_a_geq1}. The finite difference
schemes for the solution of the HJB equations, which are the main
results of this paper, are presented in Section
\ref{subsection::IntegrodifferentialScheme}. 
In Section \ref{sec: appl} these results are applied to a concrete
hedging problem in electricity markets. The proofs of the convergence
results are given in the Appendix A.

\section{The model and the quadratic hedging problem}\label{model}

Let $J$ be a Poisson random measure on $\left[0,+\infty \right)\times
\R$ defined on a filtered probability space $\left(\Omega, \mathcal F,
  \mathcal F_t,\P\right)$, $\mathcal F_t$ being the natural filtration
of $J$. We suppose that $\mathcal F_0$ contains the null sets and also
$\mathcal F=\mathcal F_T$ where $T>0$ is given. Let also $dt\times
\nu\left(dy\right)$ be the intensity measure of $J$ where $\nu$ satisfies the standard integrability condition $\int_{\R}\left(1\wedge |y|^2\right)\nu\left(dy\right)<\infty$. We denote
\begin{align*}
&\tilde J(dt\times dy ):=J(dt\times dy) - dt\times \nu(dy)   
\end{align*} 
the compensated martingale jump measure. On this probability space we introduce a family of real-valued Markov pure jump process as the solution of the following:
\begin{align}
dZ_r^{t,z}:=&\mu\left(r,Z_r^{t,z}\right)dr +\int_{\R}
\gamma\left(r,Z_{r-}^{t,z},y \right)\tilde J\left(dr\times
  dy\right),\, Z_t^{t,z}=z,\, r\geq t\label{Z}
\end{align}
for $t\in [0,T)$ and $z\in\R$. The asset price process $S$ is given by $S^{t,z}_u=\exp(Z^{t,z}_u)$. We make the following assumptions:
\begin{assumption}
\label{assumptions_1} $\, $ \\
\noindent \textbf {[C]- }\underline{\textbf {The coefficients}}.
\renewcommand{\theenumi}{\roman{enumi})}
\begin{enumerate}
\item There exists $\Maxmu\geq 0$ such that $ \left\|\mu\right\|_{\infty}\leq \Maxmu$.
\item For all $t\in[0,T]$ and $y\in\R$ the functions $z\to \mu(t,z) $ and $z\to \gamma(t,z,y)$ belong to $\mathcal C^1(\R,\R)$.
  \item There exist $\Klip\geq 0$, $\Klid \geq 0$ and a positive locally bounded function $\rho:\R \to \R^+$ such that for all $y\in\R$ and all $t\in[0,T]$ we have
\begin{align*}
&\left| \mu(t,z)-\mu(t,z')\right| \leq \Klip|z-z'|\\
&\left| \gamma(t,z,y)-\gamma(t,z',y)\right| \leq \Klid\rho(y)|z-z'|
\end{align*}
\end{enumerate}

\noindent \textbf {[L]- }\underline{\textbf {The \lev measure}}.
The \lev measures $\nu(dy)$ verifies $\nu(dy)=\nu(y)dy$ where $\nu(y):=g(y)|y|^{-(1+\alpha)}$ for some $\alpha\in(1,2)$, where $g$ is a measurable function bounded in a neighborhood of zero: 
$$
0<m_g\leq g(y) \leq M_g,\, \forall y\in (-y_0, y_0)
$$
for some positive constants $m_g,M_g$ and some $y_0>0$. We also assume that 
$$
\lim_{y\to 0^-}g(y)=g(0^-)\quad\text{and}\quad \lim_{y\to 0^+}g(y)=g(0^+)\quad \text{with $g(0^+),g(0^-)>0$}.
$$
\textbf{Example: } The tempered stable (CGMY) processes, whose L\'evy measure is of the form
$$
\nu(y) := \frac{c_{-}}{|y|^{1+\alpha}} e^{-\lambda_{-}|y|} \mathbf{1}_{y<0} + \frac{c_{+}}{|y|^{1+\alpha}} e^{-\lambda_{+}|y|} \mathbf{1}_{y>0} 
$$
for $c_{-}>0, c_{+}>0, \lambda_{-}>0$ and $\lambda_{+}>0$ satisfies
the above assumption. 

\noindent \textbf {[I]- }\underline{\textbf {Integrability conditions}}.
The function 
\begin{align*}
 \tau\left(y\right):=&\max\left(\sup_{t,z}\left(|\gamma\left(t,z,y\right)|, e^{|\gamma\left(t,z,y\right)|}-1\right),\, \rho(y)\right) 
\end{align*}
verifies
\begin{align*}
 \sup_{0<|y|\leq y_0}\frac{\tau(y)}{|y|}\leq M \,&\qquad\text{and}\qquad & \tau\in\mathbb L^{4}(\{|y|\geq y_0 \},\nu(dy))
\end{align*}
\noindent \textbf {[ND]- }\underline{\textbf {No degeneracy}}.
The function
\begin{align*}
\Gamma(y):=\inf_{t,z}\left( e^{\gamma\left(t,z,y\right)}-1\right)^2&\qquad\text{verifies}\quad& |\Gamma|:=\int_\R \Gamma^2 (y) \nu(dy)>0
\end{align*}
\noindent \textbf {[RG]- }\underline{\textbf {Regularity of the $\gamma$ function}}.
\renewcommand{\theenumi}{\roman{enumi})}
\begin{enumerate}
\item For any $t,z$ the mapping $y\to \gamma(t,z,y)$ is continuous,
  strictly increasing and maps $\mathbb R$ onto $\mathbb R$. Moreover,
  it is
 twice
  continuously differentiable on $|y|\leq y_0$ for some $y_0>0$ and there
  exist positive constants $m_1,m_2$  such that 
\begin{align*}
0<m_1\leq \inf_{t,z,|y|\leq y_0}|\gamma_y(t,z,y)| & \text{and} & \sup_{t,z,|y|\leq y_0}|\gamma_{yy}(t,z,y)|\leq m_2
\end{align*}
In particular $\gamma $ as function of $y$ is invertible: we call $\gamma^{-1}(t,z,y)$ its inverse.
\item For all $t,z\in[0,T)\times \R$, $\gamma_y(t,z,0)=1$
\item The function $\gamma_y$ is \lip continuous in the variable $z$:
$$
\sup_{t,z,|y|\leq y_0}\left| \gamma_y(t,z+h,y)-\gamma_y(t,z,y)\right| \leq m_2|h|
$$
\end{enumerate}
\noindent We denote $\Kmax:=\max(\Klip,\, \Klid)$,
\begin{align}
\tilde \mu:=& \mu+\int_{\R}(e^{\gamma}-1-\gamma)\nu(dy) \text{ and } \left\|\tilde \mu\right\|:=\sup_{t,z}|\tilde \mu(t,z)|\label{tilde_mu}
\end{align}
In the rest of the paper we denote $\left\|\tau \right\|_{1,\nu}:= \int_{|y|\geq 1}\tau(y)\nu(dy)$ whereas $\left\|\tau \right\|_{2,\nu}^2:=\int_{\R}\tau^2(y)\nu(dy)$.
\end{assumption}

It is well known that the SDE \eqref{Z} has a unique strong solution \citep{jaco.shyr.03}. 
\subsection*{On the Assumption $\mathbf[RG]-ii)$}
Among the assumptions listed above, undoubtedly $\mathbf[RG]-ii)$ seems to be
the most restrictive one: if for example the jump function is of the
form $\gamma(t,z,y)=\hat\gamma(t,z)y$ then the only possible choice
would be $\hat\gamma(t,z)=1$ for all $t,z$. In this paragraph we prove
that this assumption could be relaxed by making a special change of
variable. More precisely, given a process $Z$ which does not satisfy
this assumption, we look for a process $L$ defined by $L_t=\phi(t,Z_t)$ for some smooth function $\phi$ such that
\begin{align}
&dL_s^{t,l}:= \mu^{L}(s,L_s^{t,l}) ds + \int  \gamma^L(s,L_{s-}^{t,l},y)\bar J(dyds) \label{proc_L}
\end{align}
where $\mu^L$ and $\gamma^L$ satisfy Assumptions \ref{assumptions_1},
in particular, $\partial_y\gamma^L(t,l,0)=1$ for all $t,l$. If for such $\phi$, the function $z\to \phi(t,z)$ is invertible and smooth enough to apply \itos formula then 
\begin{align}
\gamma^L(t,l,y):=&  \phi(t,\phi^{-1}(t,l)+\gamma(t,\phi^{-1}(t,l),y))- l \label{gamma_L}\\
\nonumber
\mu^L(t,l):=&\frac{\partial\phi}{\partial t} (t,\phi^{-1}(t,l)) +\mu(t,\phi^{-1}(t,l))\frac{\partial\phi}{\partial z} (t,\phi^{-1}(t,l)) \\
+& \int_{|y|\leq 1} \left(\gamma^L(t,l,y)-\gamma(t,\phi^{-1}(t,l),y)\frac{\partial\phi}{\partial z} (t,\phi^{-1}(t,l))\right)\nu(dy)\label{mu_L}
\end{align}
In particular one has 
$$
\gamma_y^L(t,l,0) =\frac{\partial \phi}{\partial z}(t,\phi^{-1}(t,l)))\gamma_y(t,\phi^{-1}(t,l),0)
$$
If we select for example 
\begin{align}
\phi(t,z):=\int_0^{z} \frac{ds}{\gamma_y(t,s,0)} \label{phi_change}
\end{align}
then trivially $\gamma_y^L(t,l,0)=1$ for all $t,l$. The following Lemma shows that this choice guarantees that the coefficients $\mu^L$ and $\gamma ^L$ verify Assumptions \ref{assumptions_1}
\begin{lemma}
\label{lem_gammaZ_good}
Assume that there exist positive constants $m_1,m_2$ such that
\renewcommand{\theenumi}{\roman{enumi})}
\begin{enumerate}
\item For all $t,z \in[0,T]\times \R$ the mapping $y\to\gamma(t,z,y)$ is differentiable at $y=0$ and 
$$
0 < m_1 \leq |\gamma_y(t,z,0)| \leq m_2 \text{  for all $t,z \in[0,T]\times \R$ }
$$
\item The function $(t,z)\to\gamma_y(t,z,0)$ is differentiable and 
$$
\left|\frac{d}{dt}\gamma_y(t,z,0)\right|+\left|\frac{d}{dz}\gamma_y (t,z,0)\right| \leq m_2 \text{  for all $t,z \in[0,T]\times \R$ }
$$
\item The function $z\to \frac{d}{dt}\gamma_y(t,z,0)$ is \lip continuous:
$$
\left|\frac{d}{dt}\gamma_y(t,z,0)-\frac{d}{dt}\gamma_y(t,z',0)\right|\leq m_2|z-z'| \text{  for all $t\in[0,T],\, z,z'\in\R$ }
$$
\end{enumerate}
Then the functions $\mu^L$ and $\gamma^L$ defined in \eqref{gamma_L}--\eqref{mu_L} with the choice of $\phi$ given by \eqref{phi_change} verify Assumptions  \ref{assumptions_1}.
\end{lemma}
The proof of this Lemma can be found in \cite{defr.12}, Lemma 7.16.

\noindent The message of this Lemma is that, up to some regularity of
the function $\gamma_y$ at $y=0$, it is possible to remove the
assumption $\mathbf[RG]-ii)$: we could work with the process $L$
instead of $Z$ and derive all the results for $L$. By applying the
function $\phi$ we could then obtain the corresponding results for the process $Z$. We refer to Chapter 7 in \cite{defr.12} for further details. 

Nevertheless, we prefer here to work with assumption $\mathbf[RG]-ii)$ because it will make all computations easier to handle.

\subsection*{Admissible strategies and the value functions}
To describe the set of admissible strategies in the quadratic hedging problem we follow the ideas developed in \cite{cern.kall.07}: we first introduce the sets of simple strategies:
\begin{align}
\nonumber
&\mathcal D:= \left\{\theta:= \sum_{i} Y_i \mathbf 1_{[\varsigma_i,\varsigma_{i+1}}),\, Y_i\in\mathbb L^{\infty}(\mathcal F_{\varsigma_i}),\, \varsigma_i\leq \varsigma_{i+1}\textrm{ stopping times}\right\}\\
& \mathcal D_t:=\left\{\theta\mathbf 1_{(t,T]}\left|\right.\, \theta\in\mathcal D\right\}\label{D} 
\end{align}
The set of admissible strategies is a subset of the $\mathbb L^2(\P)$-closure of $\mathcal D$:
\begin{align}
&\A:=\left\{\theta\in\bar{\mathcal D}:\, \int \theta dS \in \mathbb L^2(\P)\right\}\label{set_A}
\end{align}
We define the wealth process for all $t,z,x$ by
\begin{align}
&dX_r^{t,z,x,\theta}:=\theta_{r-}dS_r^{t,z},\quad X_t^{t,z,x,\theta}:=x\label{X}
\end{align}
where $\theta$ represents the number of shares in the portfolio at time $t$. The set of admissible controls is then given by
\begin{align}
\label{At}
&\At:=\left\{ \theta\textbf{1}_{\left(t,T\right]}\left|\right.\, \theta\in\A\quad x+\int_t^{T}\theta_{r-} dS_r^{t,z}\in\mathbb L^2(\P)\right\}
\end{align}

 Consider a European option of the form $f\left(Z_T\right)$ where $f$ is, for the moment, a bounded and measurable function. The quadratic hedging problem can be formulated as follows:
\begin{eqnarray*}
\QH:&& \textrm{minimize }\E^{\P}\left[\left(f\left(Z_T^{0,z}\right)-X_T^{0,z,x,\theta}\right)^2 \right] \\
&&\textrm{over $\theta \in \mathcal X(0,z,x)$}
\end{eqnarray*}
The dynamic version of $\QH$ gives us the value function of the problem:
\begin{align}
& v_f\left(t,z,x\right):=\inf_{\theta\in\At}\E^{\P}\left[\left(f\left(Z_T^{t,z}\right)-X_T^{t,z,x,\theta}\right)^2 \right] \label{v_f_t}\\
\nonumber
& v_f\left(T,z,x\right) = \left(f\left(z\right)-x\right)^2 
\end{align}
 As remarked by several authors, the function $v_f$ has the following structure
\begin{align}
v_f(t,z,x)= a(t,z)x^2 + b(t,z)x+ c(t,z) \label{v_as_quadratic_x}
\end{align}
In particular, taking $f=0$, one has
\begin{align}
&v_0\left(t,z,x\right):= x^2\inf_{\theta\in\At}\E\left[\left(1+\int_t^T \theta_{r-} dS_r^{t,z}\right)^2\right]\label{v0}
\end{align}
because the set $\At$ is a cone. Consequently 
\begin{align}
a\left(t,z\right):=\inf_{\theta\in\Ata}\E\left[\left(1+\int_t^T \theta_{r-} dS_r^{t,z}\right)^2\right] \label{a}
\end{align}
This problem is known in the literature as the \emph{pure investment
  problem}. The dual formulation of this problem relates the function
$a$ to the so called variance optimal martingale measure \citep{cern.kall.07}. We recall here some fundamental properties on the function $a$, whose proof can be found in Chapter 5 of \cite{defr.12}.
\begin{theorem}
\label{theo_a_lip} 
Under Assumptions \ref{assumptions_1}-$\mathbf {[C,I,ND]}$ the function $a$ verifies
\begin{align*}
 e^{-C\left(T-t\right)}\leq a\left(t,z\right)\leq 1 &&\text{where} &&C:=  \frac{ \left\|\tilde\mu\right\|^2}{  \left|\Gamma\right|}
\end{align*} 
Furthermore, there exists $T^*>0$ and $\Klipa \geq 0$ such that if $T<T^*$ then
$$
|a(t,z')-a(t,z)|\leq \Klipa|z-z'| 
$$
for all $t\in[0,T]$ and $z,z'\in\R$. $T^*$ depends on $\Maxmu,\, \tau,\, C$ and $\Kmax$ defined in Assumptions \ref{assumptions_1}. Moreover $T^*\to +\infty$ when $\Kmax\to 0$ and the other constants remain fixed.
\end{theorem}
Remark that these results hold true without assuming any particular
structure of the \lev measure $\nu(dy)$. The next goal is to
characterize the functions $a,b$ and $c$ as the solutions of certain PIDEs.

\section{HJB formulation and main regularity results} \label{case_a_geq1}
\paragraph{Remarks on notation} For a function $f:[0,T]\times \R \to
\R$ we denote $\left\|f \right\|_\infty:=\sup_{t\leq T, x\in \R}
|f(t,x)|$. 
For a function $\varphi$ defined on $[0,T]\times \R$ and $k\in\mathbb N$ we denote $D^k\varphi:= \partial^{k} \varphi / \partial x^{k}$ whereas $\partial_t \varphi$ denotes the derivative in the time variable. We adopt the following convention: for any $l\in \R^+$
\begin{align*}
 l& = \lfloor l \rfloor + \{l\}^-,\, \text{where } \{l\}^- \in [0,1)& \\
 l& = \lceil l \rceil + \{l\}^+,\, \text{where } \{l\}^+ \in (0,1] &\\
\end{align*}
\noindent Let us first introduce the functional spaces in which we will work: for $\beta\in(0,1]$ we define
\begin{align*}
\langle \psi \rangle^{(\beta)}:=\sup_{x,0<|h|\leq 1} \frac{|\psi(x+h)-\psi(x)|}{|h|^\beta},&\quad&\langle \varphi \rangle^{(\beta)}_{Q_T}:= \sup_{t,x,0<|h|\leq 1}\frac{|\varphi(t,x+h)-\varphi(t,x)|}{|h|^\beta}
\end{align*}
The elliptic \hold space of order $l$, $\HZR[l]{n}$, is defined as the space of continuously differentiable functions $ \psi$ for all order $j\leq \lceil l \rceil$ with finite norm
\begin{align}
&\left\| \psi\right\|_{l}:= \sum_{j=0}^{\lceil l \rceil } \sum_{(j)}\left\| D^j_x \psi\right\|_{\infty} +\sum_{(\lceil l \rceil)}\langle D^{\lceil l \rceil}_x \psi\rangle^{( \{l\}^+)}\label{holdzyg_norm_ell}
\end{align}
where $\sum_{(j)}$ represents the summation over all possible derivative of order $j$. The parabolic \hold space $\HZpar[l]{n}$ is defined as the set of measurable functions $\varphi:[0,T]\to \HZR[l]{n}$ with finite norm
\begin{align}
&\left\|\varphi\right\|_{l} := \sum_{j=0}^{\lceil l\rceil }\sum_{(j)}\left\| D^{j}_x\varphi \right\|_{\infty} + \sum_{(\lceil l\rceil )}\langle D^{\lceil l \rceil }_x\varphi \rangle_{Q_T}^{(\{l\}^+)}\label{hold_zyg_norm_par}
\end{align}
The spaces defined above are all Banach spaces equipped with their respective norms. For a complete description see for example Chapter I in \cite{adam.four.09}. \\

\bigskip
In the spirit of HJB approach we now introduce the operators associated to the process $Z$:
\begin{definition}\label{operators}
 For a real valued function $\varphi\in \HZpar[\alpha+\delta]{}$, $\delta>0$, we define the following linear operators
\begin{align*}
 \op[A]{\varphi}(t,z) := & -\mu\frac{\partial \varphi}{\partial z}(t,z)\\
\op[B]{\varphi}(t,z)  := &  \int_\R \left(\varphi(t,z+\gamma(t,z,y))-\varphi(t,z)-\gamma(t,z,y)\frac{\partial \varphi}{\partial z}(t,z)\right)\nu(dy) 
\\
\op[Q]{\varphi}(t,z) :=& \tilde\mu\varphi(t,z) +\int_\R \left(e^{\gamma}-1\right)\left(\varphi(t,z+\gamma(t,z,y))-\varphi(t,z)\right)\nu(dy)\\
\op[G]{\varphi}(t,z) :=& \int_\R \left(e^{\gamma}-1\right)^2\varphi(t,z+\gamma(t,z,y))\nu(dy)
\end{align*}
where $\tilde\mu$ stands for $\tilde\mu(t,z)$ and so on. In addition,
$\genop[H]$ denotes the nonlinear operator
\begin{align*}
\op[H]{[\varphi]}(z):=&  \inf_{|\pi|\leq \bar\Pi}\left[2\pi \op[Q]\varphi(t,z)+\pi^2 \op[G]{\varphi}(t,z) \right]
\end{align*}
where
\begin{align}
\bar \Pi :=\frac{e^{CT}}{\left|\Gamma\right|}\max\left(\left\|\tilde\mu\right\|_\infty, 2\left(\left\|\tau\right\|_{4,\nu}^4+\left\|\tau\right\|_{2,\nu}^2\right) \right) \left(1+\Klipa \right).\label{bound_control_HJB}
\end{align}
\end{definition}

The main result concerning the functions $a$ is:
\begin{theorem}
\label{HJB_v0} Let Assumptions \ref{assumptions_1} hold true and consider $T<T^*$  as in Theorem \ref{theo_a_lip}. The function $a$ is the unique solution  of
\begin{align}
&0 =  -\frac{\partial a}{\partial t}+ \op[A]{a}-\op[B]{a}
- \op[H]{[a]}\, , &a(T,z)=1\label{QL_PIDE_a}
\end{align}
in the \hold space $\HZpar[\alpha+\delta]{}$ for $0<\delta<\alpha-1$. The function $t\mapsto a(t,z)$ is also differentiable on $(0,T)$. The optimal strategy for the stochastic control problem \eqref{v0} is
\begin{align*}
&\theta^*_t = e^{-Z_{t-}}\pi^*\left(t,Z_{t-}\right)X^{\theta^*}_{t-} ,&\quad &X_t^{\theta^*}:=x+\int_0^t \theta^*_{r-} dS_{r}
\end{align*}
where
\begin{align}
\pi^*(t,z):=-\frac{\genop[Q]{a}(t,z)}{\genop[G]{a}(t,z)}\label{pi_opt}
\end{align}
\q
\end{theorem}
For the general value function $v_f$ we have 
\begin{theorem}
\label{HJB_v}
Let $T<T^*$  as in Theorem \ref{theo_a_lip}. Let also Assumptions \ref{assumptions_1} hold true and $f\in \HZR[\alpha+\delta]{}$ for some $0<\delta<\alpha-1$. The function $v_f$ in \eqref{v_f_t} admits the decomposition
$$
v_f(t,z,x)=a(t,z)x^2 + b(t,z)x + c(t,z)
$$
where $a$ is defined in \eqref{a}, so it does not depend on $f$, and it is the unique solution in $\HZpar[\alpha+\delta]{}$ of \eqref{QL_PIDE_a}, whereas $b$ and $c$ are the unique solutions of the following linear parabolic PIDEs
\begin{align}
0=&-\frac{\partial b}{\partial t} +\genop[A]{b}-\genop[B]{b}-\pi^*\genop[Q]{b}, & \qquad b(T,.)&=-2f \label{PIDE_b};\\
0=&-\frac{\partial c}{\partial t} +\genop[A]{c}-\genop[B]{c}+\frac{1}{4}\frac{(\genop[Q]{b})^2}{ \genop[G]{a}}, & \qquad c(T,.)&=f^2 \label{PIDE_c}
\end{align}
in the \hold space $\HZpar[\alpha+\delta]{}$, where $\pi^*$ is defined
in \eqref{pi_opt}. The functions $t\mapsto a(t,.),b(t,.)$, $c(t,.)$ are also differentiable on $(0,T)$.\\ 
Furthermore the optimal policy in the control problem \eqref{v_f_t} is given by
\begin{align}
\theta^*_t:= e^{-Z_{t-}}\left( \pi^{*}(t,Z_{t-}) X^{\theta^*}_{t-} -\frac{1}{2}\frac{\genop[Q]{b}(t,Z_{t-})}{\genop[G]{a}(t,Z_{t-})}\right),\, X_t^{\theta^*}:=x+\int_0^t \theta^*_{r-} dS_{r}\label{opt_theta_vf}
\end{align}
\q
\end{theorem}
The proof of these results can be found in Chapter 7 of \cite{defr.12}. From the decomposition \eqref{v_as_quadratic_x} we also obtain the optimal price in \eqref{v_f_t}:
\begin{align}
x^*(f):= \text{arg}\inf_{x\in\R} v_f(t,z,x)=-\frac{b^f(t,z)}{2a(t,z)}\label{qh_price}
\end{align}
which is a linear function of the payoff $f$ since $b^f$ is.

\paragraph{Non smooth payoff}
Theorem \ref{HJB_v} allows us to characterize the value function $v_f$
when the payoff function $f$ is sufficiently smooth, i.e. $f\in
\HZR[\alpha+\delta]{}$. However, in most cases of interest (for
example put
options, straddles or bear spreads) this
function is not even continuously differentiable. The following lemma proves the stability of the optimal price $x^*(f)$ and the optimal hedging strategy under small perturbations of the function $f$:
\begin{lemma}
\label{lem: stab} Let $f_1,f_2$ be two measurable functions with $f_i(Z_T^{t,z})\in\mathbb L^2(\P)$ for all $t,z$, $i=1,2$. Then for any $t<T$ and $z\in\R$ 
\begin{align*}
 \left| x^*(f_1)(t,z)-x^*(f_2)(t,z)\right|\leq &a(t,z)^{-1/2}\left\|(f_1-f_2)(Z_T^{t,z})\right\|_{\mathbb L^2(\P)}\\
& \\
\left| \left(v^{f_1}-v^{f_2}\right)(t,z,x)\right|\leq &2\left( x +\left\|(f_1+f_2)(Z_T^{t,z})\right\|_{\mathbb L^2(\P)}\right) \left\|(f_1-f_2)(Z_T^{t,z})\right\|_{\mathbb L^2(\P)}
\end{align*}

 Fix now $(t,z,x)$ and let $f_n$ such that
 $\left\|(f_n-f)(Z_T^{t,z})\right\|_2\to 0$, $n\to \infty$. If
 $\theta^n$ is the optimal control in the problem \eqref{v_f_t} with payoff function $f_n$ then, for all $\varepsilon>0$, there exists some $N>0$ such that for any $n\geq N$ one has
\begin{align*}
\left|v_f(t,z,x)-\E^\P\left[\left(f(Z_T^{t,z})-x-\int_t^T\theta^n_{r-}dS_r^{t,z} \right)^2 \right] \right|\leq \varepsilon
\end{align*}
\end{lemma}
The proof of this Lemma can be found in Chapter 5 of
\cite{defr.12}. One can thus approximate a non-smooth payoff function
$f$ with smooth functions $f_n$, controlling the error on the
value function and the cost of the hedging strategy with $\left\|f-
  f_n\right\|_{2}$. Furthermore the corresponding strategies
$(\theta_n)$ become $\varepsilon$-optimal for the pay-off $f$ starting
from sufficiently large $n$.

\section{Numerical solution schemes} \label{subsection::IntegrodifferentialScheme}
We now present a numerical scheme to solve the PIDE introduced in
Section \ref{case_a_geq1} when the \lev measure $\nu$ verifies the
Assumption \ref{assumptions_1}-$\mathbf{[L]}$. From
\eqref{opt_theta_vf} and \eqref{qh_price} we remark that in order to
solve the problem \eqref{v_f_t}, i.e. to find the optimal strategy
$\theta^*$ and the optimal price $x^*$, we only need to compute the
functions $a$ and $b$, solutions, respectively, of PIDEs
\eqref{QL_PIDE_a} and \eqref{PIDE_b}. 

The finite difference discretization schemes will be constructed using the Markov chain
approximation technique developed in \cite{kush.76}. One of the advantages of the probabilistic treatment
is that it allows to estimate the error due to the truncation of the
domain in a simple manner. 
The Markov chain approximation method works as follows. We first
construct a discrete-time Markov process $(\widehat
Z_{t_{i}})_{i=0,\dots,N_T}$ evolving on a regular space grid $z_j = j\Delta
z,-N<j<N$ and regular time grid $t_i = i \Delta t$, $i=0,\dots,N_T$
with $\Delta t = \frac{T}{N_T}$, approximating the process $Z$ defined in \eqref{Z}. We
then replace the process $Z$ with the Markov chain $\widehat Z$ in the
quadratic hedging problem. The dynamic programming algorithm for the
discretized hedging problem then provides an approximation scheme for
the original control problem and thus for equations \eqref{QL_PIDE_a} and \eqref{PIDE_b}.

\subsection{Definition of the approximating Markov chain}
The
action of the generator of $Z$ on a test function $\varphi$ is given by
$$
\mathcal L^Z \phi(t,z) = \mu\frac{\partial \varphi}{\partial z}(z)+  \int_\R \left(\varphi(z+\gamma(t,z,y))-\varphi(z)-\gamma(t,z,y)\frac{\partial \varphi}{\partial z}(z)\right)\nu(dy)
$$ 
We now detail the
computation of the integral term in this generator. In order to
avoid interpolation, we use a space and time dependent grid for discretizing the L\'evy
density, with discretization points denoted by $(y_i(t,z))_{-I\leq i
  \leq I}$. We select the discretization point $y_i(t,z)$, which corresponds to
the center of $i$-th discretization interval, as the unique solution
of the equation $\gamma(t,z,y_i(t,z))= i \Delta z$. The boundaries of
the discretization intervals will then correspond to half-integer
values of $i$. Although these discretization points depend on $t$ and
$z$, we will sometimes omit this dependence to simplify notation. 

To treat the singularity of the L\'evy density at zero, we adapt
the methodology of \cite{forsyth} and divide the real line
into four disjoint regions, for an integer
$\kappa\geq 1$
\begin{eqnarray*}
 \Omega_0(t,z) & =  & \left \{ y : y_{-\kappa-\frac{1}{2}}(t,z) \le y \le y_{\kappa+\frac{1}{2}}(t,z) \right \}, \\
 \Omega_1(t,z) & = & \left \{ y : y_{\kappa+\frac{1}{2}}(t,z)< y \leq
   y_{\zeta+\frac{1}{2}}(t,z)\right\}\cup \left \{ y :
   y_{-\kappa-\frac{1}{2}}(t,z)> y \geq 
   y_{-\zeta'-\frac{1}{2}}(t,z)\right\}
,\\
 \Omega_2(t,z) & = & \left \{ y : y_{\zeta+\frac{1}{2}}(t,z)< y \leq
   y_{I+\frac{1}{2}}(t,z)\right\}\cup \left \{ y :
   y_{-\zeta'-\frac{1}{2}}(t,z)> y \geq 
   y_{-I-\frac{1}{2}}(t,z)\right\},\\
 \Omega_3(t,z) & = & \left \{ y :  y\leq y_{-I-\frac{1}{2}}(t,z) \right
\}\cup \left \{ y :  y\geq y_{I+\frac{1}{2}}(t,z) \right
\},
\end{eqnarray*}
where $\zeta = \inf\{i: y_i(t,z)\geq 1\}$ and $\zeta' = \inf\{i:
y_{-i}(t,z)\leq -1\}$. Without loss of generality, we shall always
assume that $\kappa < \zeta < I$ and $\kappa < \zeta'< I$. 

The jumps of the L\'evy measure in the regions
$\Omega_0,\dots,\Omega_3$ are treated as follows:
\begin{itemize}
\item The small jumps in the region $\Omega_0$ are truncated and the
  corresponding part of the integral operator is replaced by a local
  operator (e.g., second derivative). 
\item The jumps in the regions $\Omega_1$ and $\Omega_2$ are
  discretized. 
\item The large jumps in the region $\Omega_3$ are truncated. 
\end{itemize}
For reader's convenience the different truncation and discretization
parameters of our algorithm are listed in the following table. Theorems
\ref{maina.thm} and \ref{mainb.thm} below show how these different parameters affect the
overall approximation error. 

\medskip

\centerline{\begin{tabular}{cl}
Parameter & Meaning \\
$\Delta t$ & Time discretization step  (with $\Delta t = \frac{T}{N_T}$) \\
$\Delta z$ & Space discretization step\\
$N$ & Space grid size \\
$I$ & L\'evy measure truncation (number of points) \\
$\kappa$ & Truncation of small jumps (number of points)
\end{tabular}
}

\medskip

After replacing the small jumps with a local
operator and removing the large jumps, we obtain a generator of the form
$$
\mu\frac{\partial \varphi}{\partial z}+
\frac{D(t,z)}{2}\frac{\partial^2 \varphi}{\partial z^2}+
\int_{\Omega_1\cup \Omega_2} \left(\varphi(z+\gamma(t,z,y))-\varphi(z)-\gamma(t,z,y)\frac{\partial \varphi}{\partial z}(z)\right)\nu(dy),
$$ 
where the coefficient $D$ is defined by
\begin{align}
D(t,z):= \int_{\Omega_0} \gamma(t,z,y)^2 \nu(y)dy.\label{Dvalue}
\end{align}
Discretization of the remaining integral leads to a generator of the
form
\begin{align}
&\mu\frac{\partial \varphi}{\partial z}+
\frac{D(t,z)}{2}\frac{\partial^2 \varphi}{\partial z^2}+
\sum_{\kappa <|j|\leq I} \omega_j(t,z)
\left(\varphi(z+\gamma(t,z,y_j(t,z)))-\varphi(z)-\gamma(t,z,y_j(t,z))\frac{\partial
    \varphi}{\partial z}(z)\right),\notag\\
& = \hat\mu\frac{\partial \varphi}{\partial z}+
\frac{D(t,z)}{2}\frac{\partial^2 \varphi}{\partial z^2}+
\sum_{\kappa <|j|\leq I} \omega_j(t,z)
\left(\varphi(z+j\Delta z)-\varphi(z)\right)\label{gensum}
\end{align}
where
\begin{align*}
\hat \mu(t,z) = \mu(t,z) - \sum_{\kappa<| i| \leq I} \omega_i(t,z)
\gamma(t,z,y_i(t,z)). 
\end{align*}
and we have introduced the weights :
{\begin{equation}
 \omega_i(t,z) =\left\{\begin{array}{ll}
\displaystyle{ \frac{1}{\gamma^2(t,z,y_i(t,z))}  \int_{
    y_{i-1/2}(t,z)}^{y_{i+1/2}(t,z)} \gamma^2(t,z,y) \nu(y) dy}& \text{if  }
y_{i}\in  \Omega_1\\ & \\
\displaystyle{  \int_{y_{i-1/2}(t,z)}^{y_{i+1/2}(t,z) }\nu(y) dy} &\text{if  } y_{i}\in  \Omega_2
\\ &\\
0 & \text{otherwise  }
\end{array}\right.\label{weights.eq}
\end{equation}}

{Depending on the form of $\gamma$ and the L\'evy measure, the
integrals in $D(t,z)$ and $\omega_i(t,z)$ can often be calculated
explicitly.}
Otherwise, they may be calculated
numerically, using, e.g., the 5-point trapezoidal rule, which, in the
second case above yields
\begin{align}
 \hat \omega_i(t,z)  =   \frac{\Delta y_i}{4} \left ( \frac{1}{2} \nu( y_i - \frac{\Delta y_i}{2}) + \nu( y_i - \frac{\Delta y_i}{4}) +  \nu( y_i) + \nu( y_i + \frac{\Delta y_i}{4}) +\frac{1}{2} \nu( y_i+ \frac{\Delta y_i}{2}) \right )\label{hat_om}
\end{align}
with $ \Delta y_i(t,z) = y_{i+\frac{1}{2}}(t,z) - y_{i-\frac{1}{2}}(t,z)$.
In the following, we shall assume that these integrals are calculated
explicitly without error; the additional error introduced by the their
numerical evaluation can be easily estimated along the lines of other
computations in the following section.

Finally, to approximate the local part of the generator, introduce the
weights $\upsilon$ and $\chi$ defined by
\begin{align}
\upsilon(t,z) &=\frac{D(t,z)}{2\Delta z^2} -\frac{\hat\mu(t,z)}{2 \Delta z} \label{centra}\\
\chi(t,z) &= \frac{D(t,z)}{2\Delta z^2} +\frac{\hat\mu(t,z)}{2 \Delta z}\label{centrb}
\end{align}
if both these expressions are positive or and by
\begin{eqnarray*}
\upsilon(t,z) &= &\frac{D(t,z)}{2\Delta z^2} +\max(0,-\frac{\hat\mu(t,z)}{\Delta z}) \\
\chi(t,z) &=& \frac{D(t,z)}{2\Delta z^2} +\max(0,\frac{\hat\mu(t,z)}{\Delta z})
\end{eqnarray*}
otherwise.
Approximating the local part of the generator by central or
non-central differences in the usual way, we obtain the fully
discretized generator
\begin{multline}
\widehat{\mathcal L}^Z \varphi (z) = \chi(t,z) \varphi(z+\Delta z) + \upsilon(t,z)\varphi(z-\Delta z) -
(\chi(t,z)+\upsilon(t,z))\varphi(z) \\+ \sum_{\kappa <|j|\leq I} \omega_j(t,z)
\left(\varphi(z+j\Delta z)-\varphi(z)\right)\label{gendisc}
\end{multline}

Assume that $\Delta t$ is small enough so that for all $(t,z)$,
\begin{align}
\frac{D(t,z)}{\Delta z^2} + \frac{\Maxmu}{\Delta z} + \sum_{\kappa<| i| \leq I} |i|\omega_i(t,z) + \sum_{\kappa<| i| \leq I} \omega_i(t,z)
\leq \frac{1}{\Delta t},\label{cfl_exp}
\end{align}
where $\Maxmu$ is defined in Assumption
\ref{assumptions_1}-[\textbf{C}]-(i). Under this condition, we may introduce a discrete-time Markov chain $(\widehat
Z_{t_{i}})_{i=1,\dots,N_T}$  defined as
follows:
\begin{align*}
\widehat Z_{t_{i+1}} = \left\{\begin{aligned}
&\widehat Z_{t_i} + j \Delta z, j=\kappa+1,\dots, I, \quad \text{with
  probabilities} \quad \omega_j(t,\widehat Z_{t_i}) \Delta t \\
&\widehat Z_{t_i} + j \Delta z, j=-I,\dots, -\kappa-1, \quad \text{with
  probabilities}\quad \omega_j(t,\widehat Z_{t_i}) \Delta t\\
& \widehat Z_{t_i} +  \Delta z\quad \text{with probability}\quad
\chi(t,\widehat Z_{t_i})\Delta
t\\
& \widehat Z_{t_i} -  \Delta z\quad \text{with probability}\quad
\upsilon(t,\widehat Z_{t_i})\Delta
t\\
& \widehat Z_{t_i} \quad \text{with probability}\quad
1-\left(\upsilon(t,\widehat Z_{t_i})+\chi(t,\widehat Z_{t_i}) + \sum_{\kappa<|j|\leq I}
  \omega_j(t,\widehat Z_{t_i})\right)\Delta t\\
\end{aligned}\right.
\end{align*}

When the chain starts from the point $z$ at time $t_n$, its value will
be denoted by $\widehat Z^{z,t_n}_{t_i}$, $i\geq n$. This Markov chain
is related to the discretized generator \eqref{gendisc} in the
following way: 
$$
\mathbb E[\varphi(\widehat Z^{z,t_n}_{t_{n+1}}))] = \varphi(z) +
\Delta t \,\widehat{\mathcal L}^Z \varphi(z)
$$
For this reason,
we choose this Markov chain as the discretized approximation to the
process $Z$.

To make the
notation more compact, we shall denote the probability of transition
from $\widehat Z_{t_i}$ to $\widehat Z_{t_i} + j\Delta z$ by
$p_j(t_i,\widehat Z_{t_{i}})$, for $-I\leq j \leq I$. Also, in order
to restrict the values of the chain to the space grid, it shall be
stopped at the first moment when it exits the grid, defined by 
$$\beta^{z,t_n} = \inf\{i\geq n: \widehat
Z^{z,t_n}_{t_i} \notin [-N\Delta z, N\Delta z]\}.
$$

\subsection{A finite-difference scheme for the function $a(t,z)$} 
\label{sub:scheme_a}
The values of the approximations of the functions $a$ and $b$ at
points $(t_i,z_j)$ will be denoted, respectively, by $a^i_j$ and
$b^i_j$. We also introduce the functions $a^i(z)$ and
$b^i(z)$, defined only for $z = j \Delta z$ with $j\in \mathbb Z$, and given
by  $a^i(z_j) = a^i_j$ and $b^i(z_j) = b^i_j$. 

Consider the following discrete-time
control problem for $\widehat Z$:
$$
\hat v_f(t_n,z_j,x) = \inf_{\pi_{n+1},\dots,\pi_{N_T}}\mathbb E\left[\left(f(\widehat Z^{t_n,z_j}_{\beta^{t_n,z_j} \wedge N_T})-x\prod_{i=n+1}^{\beta^{t_n,z_j} \wedge N_T}\left(1+
\pi_i (e^{\widehat Z^{t_n,z_j}_{t_i}-\widehat Z^{t_n,z_j}_{t_{i-1}}} -1 )\right)\right)^2\right]
$$
where $\pi_i \in[-\bar \Pi, \bar \Pi]$ for $i=n+1,\dots,N_T$. This control problem is the discretized version of the original
control problem \eqref{v_f_t}. It is obtained by replacing the continuous
time process $Z = \log S$ with the discrete-time Markov chain $\widehat Z$;
assuming that the amount invested into the risky asset $\pi_t =
\theta_t S_t$ remains constant and equal to $\pi_i$ on the $i$-the
discretization interval, and stopping the Markov chain at the time $\beta^{z,t_n}$
when it first exits the space grid. Similarly to the original
continuous-time problem, it is easy to show that $\hat v_f$ has a
quadratic structure: 
\begin{align}
\hat v_f(t_n,z_j,x) = x^2 \hat a(t_n,z_j) + x\hat b(t_n,z_j) + \hat
c(t_n,z_j) \label{disc_quad.eq}
\end{align}
for some functions $\hat a$, $\hat b$ and $\hat c$ defined on the
grid. In particular, the function $\hat a$ is the solution of the pure
investment problem for the discrete process $\hat Z$ and satisfies
\begin{align}
\hat a(t_n,z_j) = \inf_{\pi_{n+1},\dots,\pi_{N_T}}\mathbb E\left[\prod_{i=n+1}^{\beta^{t_n,z_j} \wedge N_T}\left(1+
\pi_i (e^{\widehat Z^{t_n,z_j}_{t_i}-\widehat Z^{t_n,z_j}_{t_{i-1}}} -1 )\right)^2\right],\label{ahat.eq}
\end{align}
For greater generality\footnote{For example, if some a priori approximation
  for the function $a$ is available, it may be used as boundary
  condition to improve the accuracy of the scheme.}, we introduce an arbitrary boundary / terminal
condition and define the function $a^n$, approximating the function
$a$, by
$$
a^n(z_j) = \inf_{\pi_{n+1},\dots,\pi_{N_T}}\mathbb E\left[\prod_{i=n+1}^{\beta^{t_n,z_j} \wedge N_T}\left(1+
\pi_i (e^{\widehat Z^{t_n,z_j}_{t_i}-\widehat Z^{t_n,z_j}_{t_{i-1}}}
-1 )\right)^2q^a(\Delta t(\beta^{t_n,z_j}\wedge
N_T),\widehat Z^{t_n,z_j}_{\beta^{t_n,z_j} \wedge N_T})\right],
$$
where the function $q^a$ is measurable, bounded from above and below
by positive constants, and satisfies
$q^a(T,z) = 1$ for all $z$. In the numerical examples, we take
$q^a\equiv 1$ and in Theorems
\ref{maina.thm} and \ref{mainb.thm} we shall see that the effect of the boundary conditions
on the approximation becomes negligible for $\Delta z N$ sufficiently
large.

The dynamic programming principle for this discrete-time
control problem writes
\begin{align}
a^n(z_j) &= \inf_{\pi \in[-\bar \Pi, \bar \Pi]}\mathbb E\left[\left(1+
    \pi (e^{\widehat Z^{t_n,z_j}_{t_{n+1}} - z_j} -1 )\right)^2
  a^{n+1}(\widehat Z^{t_n,z_j}_{t_{n+1}})\right],\quad \text{if}\quad j \in (-N,N),\label{dpa.eq}
\\
a^n(z_j) &= q^a(t_n, z_j), \quad \text{if}\quad j \notin (-N,N),\notag
\end{align}
or in other words
\begin{align}
a^n_j = \inf_{\pi \in[-\bar \Pi, \bar \Pi]} \sum_{-I \leq l \leq I} p_l(t_n,z_j) (1+ \pi
(e^{l\Delta z} -1))^2 a^{n+1}_{j+l},\label{dpa2.eq}
\end{align}
which gives a fully explicit finite difference scheme for approximating
the function $a$. Using the explicit form of the transition
probabilities, it can be rewritten as 
\begin{align}
&\frac{a^{n+1}_j - a^{n}_j}{\Delta t} - \left(\upsilon(t_n,z_j)+\chi(t_n,z_j) \right)a^{n+1}_j + \chi(t_n,z_j)
a^{n+1}_{j+1} + \upsilon(t_n,z_j) a^{n+1}_{j-1} + \sum_{\kappa < |l| \leq
  I} \omega_l(t_n,z_j)(a^{n+1}_{j+l} - a^{n+1}_j)\notag\\
&+\inf_{\pi \in[-\bar \Pi, \bar \Pi]} \Big\{ \pi^2 \Big(\sum_{\kappa < |l| \leq I} \omega_l(t_n,z_j) 
(e^{l\Delta z} -1)^2 a^{n+1}_{j+l}+\chi(t_n,z_j) (e^{\Delta
  z} -1)^2 a^{n+1}_{j+1} +\upsilon(t_n,z_j) (e^{-\Delta
  z} -1)^2 a^{n+1}_{j-1} \Big)\notag\\ &+ 2\pi \left(\sum_{\kappa < |l| \leq I}
\omega_l(t_n,z_j) (e^{l\Delta z}-1)a^{n+1}_{j+l}+\chi(t_n,z_j) (e^{\Delta
  z} -1) a^{n+1}_{j+1} +\upsilon(t_n,z_j) (e^{-\Delta
  z} -1) a^{n+1}_{j-1} \right) \Big\}=0.\label{schemea.eq}
\end{align}
Hence, this scheme uses the following approximations for the operators
$\mathcal A$, $\mathcal B$, $\mathcal Q$ and $\mathcal G$ appearing in
Equation \eqref{QL_PIDE_a}:
\begin{align*}
&(\mathcal B_t - \mathcal A_t)a(t_n,z_j) \approx - \left(\upsilon(t_n,z_j)+\chi(t_n,z_j) \right)a^{n+1}_j + \chi(t_n,z_j)
a^{n+1}_{j+1} + \upsilon(t_n,z_j) a^{n+1}_{j-1} \\&\qquad \qquad + \sum_{\kappa < |l| \leq
  I} \omega_l(t_n,z_j)(a^{n+1}_{j+l} - a^{n+1}_j)\\
&\mathcal G_ta(t_n,z_j) \approx \sum_{\kappa <|l| \leq I} \omega_l(t_n,z_j) 
(e^{l\Delta z} -1)^2 a^{n+1}_{j+l}+\chi(t_n,z_j) (e^{\Delta
  z} -1)^2 a^{n+1}_{j+1} +\upsilon(t_n,z_j) (e^{-\Delta
  z} -1)^2 a^{n+1}_{j-1} \\
& \mathcal Q_ta(t_n,z_j) \approx \sum_{\kappa <|l| \leq I}
\omega_l(t_n,z_j) (e^{l\Delta z}-1)a^{n+1}_{j+l}+\chi(t_n,z_j) (e^{\Delta
  z} -1) a^{n+1}_{j+1} +\upsilon(t_n,z_j) (e^{-\Delta
  z} -1) a^{n+1}_{j-1}
\end{align*}

\paragraph{Stability analysis}
Under the condition \eqref{cfl_exp}, the transition
probabilities of the Markov chain are positive. Therefore, from
Equation \eqref{dpa2.eq}, by choosing $\pi=0$, it follows that whenever $a^{n+1}_j \geq 0$
for $-N < j < N$ and the boundary condition satisfies $q^a(t,z)\geq 0$
for all $t,z$, the elements of $a^{n}$ satisfy
$$
0\leq a^{n}_i \leq \max\{\max_{-N < j < N} a^{n+1}_j,
\|q^a\|_{\infty}\},\quad -N < i < N. 
$$
In other words, the scheme is $L^\infty$-stable as long as the
terminal and boundary data are non-negative. In the numerical examples
we choose $q^a(t,z)\equiv 1$ which means that $0\leq a^n_j \leq 1$ for
$n=0,\dots, N_T$ and $-N < j < N$. Therefore, condition
\eqref{cfl_exp} plays the role of the CFL condition in this setting,
ensuring the stability of the numerical scheme for $a$. 

Let us now derive a more tractable sufficient condition of stability. 
By definition of the weights $D$ and $\omega$,
\begin{align}
&\frac{D(t,z)}{\Delta z^2} + \frac{\bar \mu}{\Delta z} + \sum_{\kappa<|
  i| \leq I} |i|\omega_i(t,z) + \sum_{\kappa<| i| \leq I}
\omega_i(t,z) \\ &= \frac{1}{\Delta z^2}
\int_{\Omega_0}\gamma(t,z,y)^2 \nu(dy) + \sum_{i: y_i \in \Omega_1}
\frac{(1+|i|)}{i^2 \Delta z^2} \int_{
    y_{i-1/2}(t,z)}^{y_{i+1/2}(t,z)} \gamma^2(t,z,y) \nu(dy) \notag \\
  &\qquad \qquad + \sum_{i: y_i \in \Omega_2}
(1+|i|) \int_{
    y_{i-1/2}(t,z)}^{y_{i+1/2}(t,z)} \nu(dy)+\frac{\bar \mu}{\Delta
    z}\notag\\
&\leq \frac{1}{\Delta z^2}
\int_{\Omega_0}\gamma(t,z,y)^2 \nu(dy) + \sum_{i: y_i \in \Omega_1}
\frac{(1+|i|)(1/2+|i|)}{i^2 \Delta z} \int_{
    y_{i-1/2}(t,z)}^{y_{i+1/2}(t,z)} |\gamma(t,z,y)| \nu(dy)\notag\\
  &\qquad \qquad + \sum_{i: y_i \in \Omega_2}
\frac{1+|i|}{(|i|-1/2)\Delta z} \int_{
    y_{i-1/2}(t,z)}^{y_{i+1/2}(t,z)} |\gamma(t,z,y)|\nu(dy)+\frac{\bar \mu}{\Delta
    z}\notag\\
&\leq \frac{1}{\Delta z^2}
\int_{\Omega_0}\gamma(t,z,y)^2 \nu(dy) +\frac{2}{\Delta z}
  \frac{(\kappa+2)^2}{(\kappa+1)^2} \int_{\Omega_1 \cup \Omega_2}
|\gamma(t,z,y)|\nu(dy) +
\frac{\bar \mu}{\Delta z}\label{3term}
\end{align}
where the inequalities follow from the fact that $\gamma$ is increasing
in $y$ and $\gamma(t,z,y_{i+1/2})
= (i+1/2)\Delta z$, the fact that $\inf\{|i|: y_{i}\in\Omega_1 \} =
\kappa+1$ and some straightforward simplifications. 

We now proceed to estimate each of the above
terms. For the first term, since $\Delta z$ is small and $\kappa$ is a
constant, we may assume that $(\kappa+1/2) (1\vee \frac{1}{m_1}) \Delta z<
y_0$ with $y_0$ and $m_1$ defined in Assumption
\ref{assumptions_1}-$\mathbf{[RG]-i)}$. Then, from this assumption we
may deduce that for all $t,z$,
$$\gamma^{-1}(t,z,(\kappa+1/2)\Delta z) \leq \frac{(\kappa+1/2)\Delta
  z}{m_1}\quad \text{and} \quad\gamma^{-1}(t,z,-(\kappa+1/2)\Delta z) \geq -\frac{(\kappa+1/2)\Delta
  z}{m_1}$$ and on the other hand, $|\gamma(t,z,y)|\leq (1+m_2 y_0)|y|$
for $y\in \Omega_0$. This implies:
$$
\frac{1}{\Delta z^2}
\int_{\Omega_0}\gamma(t,z,y)^2 \nu(dy) \leq \frac{(1+m_2y_0)^2}{\Delta
z^2} \int_{|y|\leq \frac{(\kappa+1/2)\Delta z}{m_1}} y^2 \nu(dy) \leq
\frac{2M_g(1+m_2y_0)^2}{2-\alpha} \left( \frac{\kappa+\frac{1}{2}}{m_1}\right)^{2-\alpha} \Delta z^{-\alpha}
$$
The second term satisfies
\begin{align*}
&\frac{1}{\Delta z} \int_{\Omega_1 \cup \Omega_2}
|\gamma(t,z,y)|\nu(dy) \leq \frac{1}{ \Delta z}
\int_{(\Omega_1 \cup \Omega_2)\cap \{|y|\leq y_0\}}
|\gamma(t,z,y)|\nu(dy) + \frac{1}{\Delta
  z}\int_{|y|> y_0} \sup_{t,z}|\gamma(t,z,y)| \nu(dy)\\
&\leq \frac{1+m_2y_0}{\Delta z}
\int_{|y|>\frac{(\kappa+1/2)\Delta
  z}{1+m_2 y_0}}
|y|\nu(dy) + \frac{1}{\Delta
  z}\int_{|y|> y_0} \sup_{t,z}|\gamma(t,z,y)| \nu(dy) \\
&\leq \frac{(\kappa+\frac{1}{2})M_g}{ \alpha-1}
\left(\frac{\kappa+1/2}{1+m_2 y_0}\right)^{-\alpha} \Delta z^{-\alpha}+  \frac{1}{ \Delta
  z}\int_{|y|> y_0} \sup_{t,z}|\gamma(t,z,y)| \nu(dy),
\end{align*}
where we remark that the integral in the right-hand side is finite by
Assumption \ref{assumptions_1}-$\mathbf{[I]}$.

Finally, we conclude that the stability condition \eqref{cfl_exp} is
implied by the bound
\begin{align}
&\Delta t \leq \frac{\Delta z^{\alpha}}{C_1 + C_2 \Delta z^{\alpha-1}}\quad \text{with}\label{dtbound}\\
&C_1 = \frac{2M_g(1+m_2y_0)^2}{2-\alpha} \left(
  \frac{\kappa+\frac{1}{2}}{m_1}\right)^{2-\alpha} +\frac{2(\kappa+2)^2(\kappa+\frac{1}{2})M_g}{(\kappa+1)^2 (\alpha-1)}
\left(\frac{\kappa+1/2}{1+m_2 y_0}\right)^{-\alpha},\notag\\
&C_2 = \bar \mu + 2\frac{(\kappa+2)^2}{(\kappa+1)^2}\int_{|y|> y_0}
\sup_{t,z}|\gamma(t,z,y)| \nu(dy).\notag
\end{align}
When $\Delta z$ is sufficiently small, to satisfy this condition, it is enough to take $\Delta t
\leq \frac{\Delta z^{\alpha}}{C_1 + \varepsilon}$ for some small
value $\varepsilon>0$. 

In practice, \eqref{dtbound} is only an upper bound and not a
necessary condition for stability. We recommend using the sharper condition
\eqref{cfl_exp}, for example, by decreasing the time step in the
computational procedure whenever this condition is violated.

\paragraph{Accuracy analysis} The following theorem analyzes the error and convergence of our approximation
algorithm.

\begin{theorem}\label{maina.thm}
 Assume that the following technical conditions hold in addition to Assumptions \ref{assumptions_1}:
\begin{enumerate}
\item The L\'evy density $\nu$ is twice continuously differentiable with
  bounded derivatives outside any neighborhood of zero;
\item The function $g$ has bounded derivatives near zero;
\item Assumption
\ref{assumptions_1}-$\mathbf{[RG]-i)}$ holds true with $y_0=+\infty$;
\item The function $\gamma(t,z,y)$ is 3 times continuously
  differentiable with respect to $y$ with bounded derivatives;
\item The data of the problem are such that the functions $a$ and $b$
  together with their derivatives up to
  order 4 with respect to $z$ and up to order 2 with respect to $t$ are
  bounded. 
\end{enumerate}  

Let the condition \eqref{cfl_exp} be satisfied and assume that $N
\Delta z > I\Delta z + \bar
  \mu T$. Then, there exist two positive constants $c,C<\infty$, which
  do not depend on truncation or discretization
parameters, such that for all $\kappa > c$, 
\begin{align*}
|a_n(z_j) -a(t_n,z_j)| &\leq C \Bigg\{\frac{(1+|z_j|)}{(N - I)\Delta z - \bar
  \mu T} + 
\Delta t + \Delta 
z^{3-\alpha}\kappa^{3-\alpha} + \Delta z^{3-\alpha}\kappa^{1-\alpha} \\ &+  \int_{|y|\geq \frac{I
    \Delta z}{\|\gamma'_y\|_\infty}} (1+|y| + \tau(y) +
\tau^2(y))\nu(dy)\Bigg\},
\end{align*}
 for $-N \leq j \leq N$ and $0\leq n \leq N_T$;
\end{theorem}
The proof of this result is given in Appendix A. 
\begin{remark}
This theorem gives the decomposition of the global approximation error
into four sources: the domain truncation error (first term in brackets), the time
discretization error (second term), the error coming from space
discretization and truncation of small jumps
(third and fourth term) and the big jump truncation (last term). It
may be used for the optimal choice of the parameters of our
algorithm.
\end{remark}
\begin{remark}
The constant $c$ is needed to ensure that for $\Delta z$ sufficiently
small, the weights $\upsilon$ and $\chi$ may be assumed to be given by \eqref{centra}-\eqref{centrb}.
If the L\'evy measure is locally symmetric near zero, that  is, $c_+ =
c_-$, one can take $c=1$.  Indeed, in this case, the integral
$\int_{\varepsilon < x \leq 1} x \nu(dx)$ remains bounded as
$\varepsilon \to 0$, which means that $\hat \mu(t,z)$ remains bounded
as well. On the other hand, $\frac{D(t,z_j)}{\Delta z^2}$ is of order of $(\kappa \Delta
z)^{-\alpha}$, which means that, eventually $\frac{D(t,z_j)}{\Delta
  z^2}$ becomes bigger than $\frac{\hat \mu(t,z)}{\Delta z}$ for all
$\kappa \geq 1$ (since under our assumptions $\alpha>1$). 
\end{remark}

\begin{remark}[Implicit-explicit scheme]\label{impexp.rem}
In practice, it may be preferable to use an implicit scheme for  the
convection diffusion part and an explicit one for the integral part:
\begin{align}
&\frac{a^{n+1}_j - a^{n}_j}{\Delta t} - \left(\upsilon(t_n,z_j)+\chi(t_n,z_j) \right)a^{n}_j + \chi(t_n,z_j)
a^{n}_{j+1} + \upsilon(t_n,z_j) a^{n}_{j-1} + \sum_{\kappa <|l| \leq
  I} \omega_l(t_n,z_j)(a^{n+1}_{j+l} - a^{n+1}_j)\notag\\
&+\inf_{\pi \in[-\bar \Pi, \bar \Pi]} \Big\{ \pi^2 \Big(\sum_{\kappa < |l| \leq I} \omega_l(t_n,z_j) 
(e^{l\Delta z} -1)^2 a^{n+1}_{j+l}+\chi(t_n,z_j) (e^{\Delta
  z} -1)^2 a^{n+1}_{j+1} +\upsilon(t_n,z_j) (e^{-\Delta
  z} -1)^2 a^{n+1}_{j-1} \Big)\notag\\ &+ 2\pi \Big(\sum_{\kappa < |l| \leq I}
\omega_l(t_n,z_j) (e^{l\Delta z}-1)a^{n+1}_{j+l}+\chi(t_n,z_j) (e^{\Delta
  z} -1) a^{n+1}_{j+1} +\upsilon(t_n,z_j) (e^{-\Delta
  z} -1) a^{n+1}_{j-1} \Big) \Big\}=0.\label{schemeaimp.eq}
\end{align}
Assume that $\Delta t$ is small enough so that for all $(t,z)$,
\begin{align}
 \sum_{\kappa<| i| \leq I} \omega_i(t,z)
\leq \frac{1}{\Delta t}.\label{cfl_imp}
\end{align}
Taking $\pi=0$ in the expression to be minimized, we get
$$\frac{a^{n+1}_j - a^{n}_j}{\Delta t} - \left(\upsilon(t_n,z_j)+\chi(t_n,z_j) \right)a^{n}_j + \chi(t_n,z_j)
a^{n}_{j+1} + \upsilon(t_n,z_j) a^{n}_{j-1} + \sum_{\kappa <|l| \leq
  I} \omega_l(t_n,z_j)(a^{n+1}_{j+l} - a^{n+1}_j)\geq 0,
$$
or equivalently
$$
a^n_j \leq a^{n+1}_j \frac{1 - \Delta t \sum_{\kappa <|l| \leq
  I} \omega_l(t_n,z_j)}{1+\Delta t(\chi + \upsilon)} + 
\frac{\Delta t \chi a^n_{j+1}}{1+\Delta t (\chi + \upsilon)} + 
\frac{\Delta t \upsilon a^n_{j-1}}{1+\Delta t (\chi + \upsilon)} + \frac{\Delta t \sum_{\kappa <|l| \leq
  I} \omega_l(t_n,z_j)a^{n+1}_{j+l}}{1+\Delta t(\chi + \upsilon)}
$$
Under \eqref{cfl_imp} this implies
$$
a^n_j \leq \frac{\|a^{n+1}\|_\infty \vee \|q^a\|_\infty}{1+\Delta t(\chi + \upsilon)} + 
\frac{\Delta t \chi a^n_{j+1}}{1+\Delta t (\chi + \upsilon)} + 
\frac{\Delta t \upsilon a^n_{j-1}}{1+\Delta t (\chi + \upsilon)},
$$
which means that 
$$
a^n_j \leq \|a^{n+1}\|_\infty \vee \|q^a\|_\infty,\quad \text{all $j$.}
$$
This means that the implicit-explicit scheme is $L^\infty$-stable if
the boundary condition is non-negative, and the non-negativity of the
solution is imposed at each step of the scheme. 

From the above argument we see that \eqref{cfl_imp} plays the role of
the CFL condition for the implicit-explicit scheme. An a priori upper
bound for $\Delta t$ in terms of model parameters is therefore given
by the bound on the second term in \eqref{3term}. We conclude that the
stability of the implicit-explicit scheme is guaranteed by the
condition 
\begin{align*}
&\Delta t \leq \frac{\Delta z^{\alpha}}{C_1 + C_2 \Delta z^{\alpha-1}}\quad \text{with}\\
&C_1 = \frac{2(\kappa+2)^2(\kappa+\frac{1}{2})M_g}{(\kappa+1)^2 (\alpha-1)}
\left(\frac{\kappa+1/2}{1+m_2 y_0}\right)^{-\alpha},\notag\\
&C_2 =  2\frac{(\kappa+2)^2}{(\kappa+1)^2}\int_{|y|> y_0}
\sup_{t,z}|\gamma(t,z,y)| \nu(dy).\notag
\end{align*}
When $\Delta z$ is sufficiently small, to satisfy this condition, it is enough to take $\Delta t
\leq \frac{\Delta z^{\alpha}}{C_1 + \varepsilon}$ for some small
value $\varepsilon>0$. The stability condition for the
implicit-explicit scheme is therefore similar to the one for the fully
explicit scheme, but with a smaller constant $C_1$. 
\end{remark}
\paragraph{Computing the optimal strategy} The optimal hedging (pure
investment) strategy may be computed as the value of the maximizer in
\eqref{dpa.eq}:
with 
$$
\pi^{*n}_j \equiv\pi^{*n}(z_j)= -\bar \Pi\vee \left(-\frac{\mathbb
  E\left[\left(e^{\widehat
        Z^{t_n,z_j}_{t_{n+1}}-z_j}-1\right)a_{n+1}(\widehat Z^{t_n,z_j}_{t_{n+1}})\right]}{\mathbb
  E\left[\left(e^{\widehat
        Z^{t_n,z_j}_{t_{n+1}}-z_j}-1\right)^2a_{n+1}(\widehat Z^{t_n,z_j}_{t_{n+1}})\right]}\right)\wedge
\bar \Pi,
$$
The error of evaluating the hedging strategy may be greater than that
of the evaluating the value function. The following corollary gives an
estimate of this error. 
\begin{corollary}\label{hedging.cor}
Under the assumptions of Theorem \ref{maina.thm}, 
\begin{align}
|\pi^{*n}(z_j) - \pi(t_{n-1},z_j)|&\leq \frac{C}{\sqrt{\Delta t}}\Bigg\{\frac{(1+|z_j|)}{(N - I)\Delta z - \bar
  \mu T} + 
\Delta t + \Delta 
z^{3-\alpha}\kappa^{3-\alpha} + \Delta z^{3-\alpha}\kappa^{1-\alpha} \\ &+  \int_{|y|\geq \frac{I
    \Delta z}{\|\gamma'_y\|_\infty}} (1+|y| + \tau(y) +
\tau^2(y))\nu(dy)\Bigg\}.\label{defpin}
\end{align}
\end{corollary}
The proof of this Corollary can be found in Appendix A. 

\begin{remark}
For the error of approximating the optimal strategy to tend to zero,
the space discretization step $\Delta z$ must therefore be
sufficiently small compared to $\Delta t$. At the same time the CFL
condition imposes a lower bound on $\Delta z$, which may not tend to
zero faster than $\Delta t^{1/\alpha}$. Letting $\Delta z \sim \Delta
t^{1/\alpha}$, we get a convergence rate of $\Delta
t^{\frac{1}{2}\wedge (\frac{3}{\alpha} - \frac{3}{2})}$ for the
optimal strategy. 
\end{remark} 
\subsection{A finite-difference scheme for the function $b$}\label{sub:scheme_b}

The function $\hat b$ solution appearing in the representation
\eqref{disc_quad.eq} for the discretized control problem is given by
$$
\hat b(t_n,z_j) = -2\mathbb E\left[\prod_{i=n+1}^{\beta^{t_n,z_j} \wedge N_T}\left(1+
\pi^*_i (e^{\widehat Z^{t_n,z_j}_{t_i}-\widehat Z^{t_n,z_j}_{t_{i-1}}}
-1 )\right)f(\widehat Z^{t_n,z_j}_{\beta^{t_n,z_j}\wedge N_T})\right],
$$
where $\pi^*$ is the optimal strategy for \eqref{ahat.eq}. 

For greater generality we introduce an arbitrary boundary / terminal
condition and define the function $b^n$, approximating the function
$b$, by
\begin{align}
b^n(z_j) = \mathbb E\left[\prod_{i=n+1}^{\beta^{t_n,z_j} \wedge N_T}\left(1+
\pi^{*n}(\widehat Z^{t_n,z_j}_{t_i}) (e^{\widehat
  Z^{t_n,z_j}_{t_i}-\widehat Z^{t_n,z_j}_{t_{i-1}}} -1
)\right)q^b(\Delta t(\beta^{t_n,z_j}\wedge
N_T),\widehat Z^{t_n,z_j}_{\beta^{t_n,z_j} \wedge N_T})\right],\label{discb.eq}
\end{align}
where the function $q^b$ is measurable bounded and satisfies
$q^b(T,z) = -2f(z)$. 

The finite-difference approximation for the function $b$ is therefore
given by the
solution of the following linear dynamic programming problem:
$$
b^n(z_j) = \mathbb E\left[\left(1+\pi^{*n}(z_j)(e^{\widehat
      Z^{t_n,z_j}_{t_{n+1}}-z_j}-1)\right)b_{n+1}(\widehat Z^{t_n,z_j}_{t_{n+1}})\right]
$$
 for $j\in (-N,N)$ and $b^n(z_j) =
q^b(t_n,z_j)$ for $j\notin(-N,N)$. In other words, 
\begin{align}
b^n_j = \sum_{-I\leq l \leq I} p_l(t_n,z_j)(1 +
\pi^{*n}_j (e^{l\Delta z}-1)) b^{n+1}_{j+l}.\label{expbprob}
\end{align}

The terminal condition is given by $b^{N_T}_j = -2f(z_j)$ and in the
numerical examples we take $q^b(t,z) = -2f(z_j)$ as well. 

Equation \eqref{expbprob} defines a fully explicit finite difference
scheme for computing the function $b$, which can be rewritten as
\begin{align}
&\frac{b^{n+1}_j - b^{n}_j}{\Delta t} -
(\upsilon(t_n,z_j)+\chi(t_n,z_j))b^{n+1}_j +
\chi(t_n,z_j)b^{n+1}_{j+1} + \upsilon(t_n,z_j)b^{n+1}_{j-1} \notag\\ &+
\sum_{\kappa < |l| \leq I} \omega_l(t_n,z_j)(b^{n+1}_{j+l}-b^{n+1}_j) \notag\\
&+\pi^{*n}_j \left(\sum_{\kappa < |l| \leq I} \omega_l(t_n,z_j) (e^{l\Delta
  z}-1) b^{n+1}_{j+l} + \chi(t_n,z_j) (e^{\Delta z}-1) b^{n+1}_{j+1}
+ \upsilon(t_n,z_j)(e^{-\Delta z}-1) b^{n+1}_{j-1}\right)=0. \label{schemeb.eq}
\end{align}
This scheme uses the same approximations for the operators appearing in
\eqref{PIDE_b} as the scheme for the function $a$ defined in section
\ref{sub:scheme_a}. An implicit-explicit scheme for the function $b$
can be defined along the lines of Remark \ref{impexp.rem}. 

\paragraph{Stability analysis}
The numerical scheme for $b$ is $L^\infty$-stable under the condition
\eqref{cfl_exp}. Indeed, under this condition, and using the
Cauchy-Schwarz inequality, it follows from the representation
\eqref{discb.eq} that for every $n$,
\begin{align*}
&\|b^n\|_\infty \leq \|q^b\|_\infty \mathbb E\left[\prod_{i=n+1}^{\beta^{t_n,z_j} \wedge N_T}\left(1+
\pi^{*n}(\widehat Z^{t_n,z_j}_{t_i}) (e^{\widehat
  Z^{t_n,z_j}_{t_i}-\widehat Z^{t_n,z_j}_{t_{i-1}}} -1
)\right)^2\right]^{\frac{1}{2}}\\
&\leq  \frac{\|q^b\|_\infty}{(\min_{t,z}q^a(t,z))^{\frac{1}{2}}} \mathbb E\Big[\prod_{i=n+1}^{\beta^{t_n,z_j} \wedge N_T}\Big(1+
\pi^{*n}(\widehat Z^{t_n,z_j}_{t_i}) (e^{\widehat
  Z^{t_n,z_j}_{t_i}-\widehat Z^{t_n,z_j}_{t_{i-1}}} -1
)\Big)^2\\ &\qquad \qquad \times q^a(\Delta t(\beta^{t_n,z_j}\wedge
N_T),\widehat Z^{t_n,z_j}_{\beta^{t_n,z_j} \wedge
  N_T})\Big]^{\frac{1}{2}}\\
&=\|q^b\|_\infty
\left(\frac{a^n_j}{\min_{t,z}q^a(t,z)}\right)^{\frac{1}{2}}\leq \|q^b\|_\infty
\left(\frac{\|q^a\|_\infty}{\min_{t,z}q^a(t,z)}\right)^{\frac{1}{2}},
\end{align*}
by definition of $\pi^*$ and $a^n_j$. With the boundary condition $q^a\equiv 1$ for the
function $a$, one simply has $\|b^n\|_\infty \leq \|q^b\|_\infty$. 
\paragraph{Accuracy analysis}
Since the numerical scheme for the function $b$ uses the approximate
pure investment strategy $\pi^*$, the error of approximating $b$ is
determined by the error of approximating the optimal strategy. 
\begin{theorem}\label{mainb.thm}
Under the assumptions of Theorem \ref{maina.thm}, 
\begin{align*}
|b_n(z_j) - b(t_{n},z_j)|&\leq \frac{C}{\sqrt{\Delta t}}\Bigg\{\frac{(1+|z_j|)}{(N - I)\Delta z - \bar
  \mu T} + 
\Delta t + \Delta 
z^{3-\alpha}\kappa^{3-\alpha} + \Delta z^{3-\alpha}\kappa^{1-\alpha} \\ &+  \int_{|y|\geq \frac{I
    \Delta z}{\|\gamma'_y\|_\infty}} (1+|y| + \tau(y) +
\tau^2(y))\nu(dy)\Bigg\}.
\end{align*}
\end{theorem}
The proof of this result is provided in Appendix A. 

\section{Application to the electricity market}\label{sec: appl}
Many studies have shown that price spikes in electricity and gas
markets are incompatible with Gaussian dynamics
\citep{geman.roncoroni.06,mbt.08} and several models based on
L\'evy processes have been
developed to fit the observed fat tails \citep{deng.jiang.05,benth.al.07}. Most of them model the price under the martingale measure, or 
make some assumption on the change of probability resulting in a
similar model under the martingale measure
\citep{benth.al.07}. However, these martingale models are not adapted
for the evaluation of hedging strategies, since the hedging error
should be computed under the historical measure. 

In this paper, we propose a model which describes the deformation of
the forward curve directly under the historical probability, and satisfies Assumptions \ref{assumptions_1} so that we can apply Theorems \ref{HJB_v0}--\ref{HJB_v}. We start by introducing a \lev process $\hat L$ as follows
\begin{align}
\hat L_s=\zeta s + \int_0^s\int_\R y\tilde J(ds\times dy) \label{levy_jump}
\end{align}
where $\zeta\in\R$ and $\tilde J$ is a compensated Poisson random measure, whose \lev measure is denoted by $\nu(dy)$. Fix $c\in\R^+$, $l(s)=e^{-c s}$ and 
\begin{align}
A_t := \int_0^t e^{c s} d\hat L_s\label{spot_price_F}
\end{align}

We model the price at time $t$ of the future contract
with maturity $T$ and instantaneous delivery as a random perturbation
of the initial forward curve $\psi$ which is supposed to be known. Using the above notation we have
$$
\bar F_{0,T,t}=\psi(0,T) e^{l(T)A_t}
$$
By no arbitrage in the futures market, the price at time $t$ of a
future contract with duration $d$ is equal to the average over the
time period $[T,T+d]$ of the future contract prices with instantaneous
delivery. We therefore model the price at time $t$ of a future contract with delivery time $T$ and duration $d>0$ by  
$$
F_{d,T,t}= \frac{1}{d}\int_T^{T+d} \bar F_{0,s,t}ds = \frac{1}{d}\int_T^{T+d} \psi(0,s) e^{l(s)A_t}ds
$$
For reasons which will become clear in the sequel, we prefer the following notation:
\begin{align}
F_{d,T,t}:= \exp(\Phi(A_t))\quad \textrm{   where  }\quad \Phi(A):=\log\left(\frac{1}{d}\int_T^{T+d} \psi(0,s) e^{l(s)A} ds\right) \label{future}
\end{align}
The model \eqref{future} essentially states that the price of a future
contract $F_{d,T,t}$ is the average price on the interval $[T,T+d]$ of
the future contract with instantaneous delivery up to the random
perturbation $e^{l(s)A}$. In this context, the problem of hedging a European
option on $F_{d,T,t}$ with the quadratic hedging approach becomes
\begin{align}
&\mathbf{minimize}\quad \E\left[\left(  H(F_{d,T,t}) - x -\int_t^T \theta_{u-}dF_{d,T,u}\right)^2 \right] \quad\textrm{over $\theta$ and $x\in\R$} \label{elec_prob_1}
\end{align}
for a given map $H$. The process $F_{d,T,t}$ corresponds to $S$ in the
formulation \eqref{quadratic_error}. The following results proves that
$Z=\log(F)$ is a Markov jump process satisfying our assumptions.

\begin{lemma}
\label{lem_transf}
The process $Z_t:= \log(F_{d,T,t})$ verifies 
$$
dZ_t = \mu(t,Z_t)dt + \int \gamma(t,Z_{t-},y) \tilde J(dydt)
$$
where
\begin{align*}
\gamma(t,z,y):= & \Phi( \Phi^{-1}(z) + ye^{c t}) -z\\
\mu(t,z):= &\zeta e^{c t} \Phi'( \Phi^{-1}(z) ) + \int_{\R}\left(\gamma(t,z,y)-ye^{c t}\Phi'( \Phi^{-1}(z) ) \right)\nu(dy)\\
\end{align*}
Assume that the \lev measure $\nu(dy)$ is given by $\nu(dy)= g(y) |y|^{-(1+\alpha)}dy$, for some $\alpha\in(1,2)$ and a bounded, strictly positive and measurable $g$ such that the following conditions hold true:
\renewcommand{\theenumi}{\roman{enumi})}
\begin{enumerate}
\item There exists $m<\infty$ such that for all $y,y'\in (-y_0,0)\cup(0,y_0)$ with $yy'>0$, $|g(y)-g(y')|\leq m |y-y'|$ \\
\item $\displaystyle{\lim_{y\to 0^-}g(y)=g(0^-)\quad\text{and}\quad \lim_{y\to 0^+}g(y)=g(0^+)\quad \text{with $g(0^+),g(0^-)>0$}}$ \\
\item 
$\displaystyle{
\int_{y\leq -1}  y^4 \nu(dy) + \int_{1<y} e^{4y}\nu(dy) <+\infty}$
\end{enumerate}
Then the functions $\mu$ and $\gamma$ verify the Assumptions \ref{assumptions_1}-$\mathbf{[C,L,I,ND, RG_i, RG_{iii}]}$, where the function $\tau$ is given by
$$
 \tau(y)=\max \left( |y|, \left| e^y-1\right| \right)
$$
\end{lemma}
The proof of this result is given in Appendix B.

In order to apply our results (Theorems \ref{HJB_v0} and \ref{HJB_v}) we also need to verify Assumption \ref{assumptions_1}-$\mathbf{[RG_{ii}]}$ and it is easy to see that the function $\gamma$ does not verify it: however, as we have already said in Section \ref{model}, this can be avoided by making a change of variable $L_t=\phi(t,Z_t)$. We refer to Chapter 7 and 8 in \cite{defr.12} for further details.\\

In terms of the process $Z$, the problem \eqref{elec_prob_1} becomes
\begin{align}
v_f(t,z,x)=\inf_{\theta\in\At}\E\left[\left(  f(Z_T^{t,z}) - x -\int_t^T \theta_{u-}d \exp(Z_u^{t,z})\right)^2 \right] \label{elec_prob_2}
\end{align}
where $\At$ is defined in \eqref{At} and $f(z)=H(e^z)$.

We now describe a special class of pay-offs which is of interest in
problem \eqref{elec_prob_2}. Let  $p(x):=(K-x)^+$ and define
$$
h(A):= \frac{1}{d'}\int_T^{T+d'}\psi(0,s) e^{l(s) A}ds
$$
for $d'\neq d$. From \eqref{future} it follows that $h\circ \Phi^{-1}(Z_t)= F_{d',T,t}$, and then, by defining $f:=p( h\circ \Phi^{-1})$, we obtain $f(Z_t)=(K-F_{d',T,t})^+$, which is a put option written on a future contract with different duration $d'$. Using this specific option  we can rewrite problem \eqref{elec_prob_2} as follows
\begin{align*}
\inf_{\theta\in\At}\E^{t,z,x}\left[\left( (K-F_{d',T,t})^+ - x -\int_t^T \theta_{u-}dF_{d,T,u}\right)^2 \right]
\end{align*}
 The financial meaning of the above problem is particularly
 interesting: one tries to hedge (in the quadratic sense) a put option
 written on a future contract with duration $d'\neq d$ using as
 hedging instrument the future contract with duration $d$. This may be
 useful when, for example, one sells a future contract with a
 non-standardized duration in the OTC market and hedges the resulting
 position using instruments which are liquidly traded. 

\subsection{Numerical Example: the CGMY model} \label{subsec: CGMY}
In this section we study the problem  \eqref{elec_prob_2} when $L$ in \eqref{levy_jump} is a CGMY process (\cite{carr_geman_madan_yor})
\bigskip
 We can write then
$$
L_t= \left(\mu  + C \Gamma(1-Y) (M^{Y-1}- G^{Y-1}) \right)t + \int_0^t \int_\R y \tilde J(dyds)
$$
where  $\tilde J$ is a compensated Poisson random measure with
intensity 
$$\nu(dy) = \nu(y) dy,\quad \nu(y)= C \frac{ e^{-My}}{y^{1+Y}} 1_{y>0} + C \frac{e^{Gy}}{|y|^{1+Y}} 1_{y<0} ,
$$ 
The goal of this paragraph is to solve numerically the equations
 \eqref{QL_PIDE_a} and \eqref{PIDE_b} when the model for $Z$ is given in
 Lemma \ref{lem_transf} and the source of randomness $L$ in
 \eqref{levy_jump} is given by the CGMY process introduced above. We
 first apply the implicit-explicit variant \eqref{schemeaimp.eq} of the
 scheme \eqref{schemea.eq} for the function $a$ with maturity $T =
 7$. The coefficients $\mu$ and $\gamma$ given in Lemma
 \ref{lem_transf}:
\begin{align*}
\mu(t,z) = & \Phi^{'}( \Phi^{-1}(z))  \left(\mu  + C \Gamma(1-Y) (M^{Y-1}- G^{Y-1})\right) e^{c t} \\
+& \int_{\R} \left[\Phi\left( \Phi^{-1}(z)+ ye^{ct}\right) - z -  ye^{c t} \Phi^{'}( \Phi^{-1}(z))\right] \nu(dy) \\
\gamma(t,z,y)  = & \Phi\left( \Phi^{-1}(z)+ ye^{c t}\right) - z, 
\end{align*} and the function $D(t,z)$ introduced in \eqref{Dvalue}, are computed
numerically. To this end, we compute the integration points $y_i$, as in Section \ref{subsection::IntegrodifferentialScheme}, such that $\gamma(t,z,y_i(t,z))=i\Delta z$, or equivalently
$$
y_i(t,z):=e^{-c t} \left(\Phi^{-1}(z+i\Delta z)-\Phi^{-1}(z) \right)
$$
We choose $k = 0$ in the definition of $\Omega_0(t,z)$, and by expanding $\gamma$ around zero, we obtain
\begin{align*}
D(t,z):=& \int_{y_{-k-1/2}(t,z)}^{y_{k+1/2}(t,z)}\gamma(t,z,y)^2 \nu(dy) \simeq e^{2c t}(\Phi^{'}(\Phi^{-1}(z)))^2 \int_{y_{-k-1/2}(t,z)}^{y_{k+1/2}(t,z)} y^2 \nu(dy) \\
\simeq & e^{c t} \frac{C}{2-Y} \left((\Phi^{-1}(z+ \frac{1}{2}\Delta z)- \Phi^{-1}(z))^{2-Y} + (\Phi^{-1}(z) - \Phi^{-1}(z- \frac{1}{2}\Delta z))^{2-Y} \right)
\end{align*}
Using an approach similar to the one of Section \ref{subsection::IntegrodifferentialScheme}, we approximate $$\int_{\R} \left[\Phi\left( \Phi^{-1}(z)+ ye^{c t}\right) - z - ye^{c t} \Phi^{'}( \Phi^{-1}(z))\right] \nu(dy)$$ with
\begin{align*}
& \frac{1}{2} \Delta z (\Phi^{'}( \Phi^{-1}(z)) e^{c t})^2 \frac{\delta}{\pi}
 +\sum_{y_i,\, |i|>k} \hat\omega_i(t,z)\left[\Phi\left( \Phi^{-1}(z)+ y_ie^{c t}\right) - z -  y_ie^{c t} \Phi^{'}( \Phi^{-1}(z))\right]  \\
=
 & \frac{1}{2} \Delta z (\Phi^{'}( \Phi^{-1}(z)) e^{c t})^2 \frac{\delta}{\pi}
+\sum_{y_i,\, |i|>k} \hat\omega_i(t,z)\left[i\Delta z - \left(\Phi^{-1}(z+i\Delta z)-\Phi^{-1}(z) \right) \Phi^{'}( \Phi^{-1}(z))\right]
\end{align*}
where the weights $\hat\omega_i(t,z)$ are given in \eqref{hat_om}. The
effect of approximating the coefficient of the PIDE on the solution of
the original problem, appearing for example in the pricing of European
options, has been studied in \cite{jacobsen1}.

We solve the problem \eqref{elec_prob_2} for European options $f$, with maturity one week and delivery for the 7 days of the week. We recall that the future contract in this case is given by
\begin{align}
F_{7days,1week,t}=\frac{1}{7}\int_{7}^{14}\psi(0,s) e^{l(s)A_t} ds\label{future_7}
\end{align}
and $A_t$ is given in \eqref{spot_price_F} with $\hat L$ being the
CGMY process defined above. The initial forward curve for the seven days of delivery is given in Table \ref{tab: forw_curve}, indicating that prices are lower for the week end.
\begin{table}[!h]
\centering
\begin{tabular}{|l|l|c|}
\hline
Day  & $s$ & Price  $(\psi(0,s))$ \\
\hline
Monday &$s\in[7,8)$ &80 \\
Tuesday & $s\in[8,9)$&90 \\
Wednesday &$s\in[9,10)$ &70 \\
Thursday &$s\in[10,11)$&90 \\
Friday & $s\in[11,12)$&80 \\
Saturday & $s\in[12,13)$&70 \\
Sunday &$s\in[13,14]$ &60 \\
\hline
\end{tabular}  
\caption{The forward curve. Prices are given in Eur.}
\label{tab: forw_curve}
\end{table}
This case with continuous long delivery
corresponds to a non stationary process where the hedge cannot be
calculated efficiently as in \cite{huba.kall.kraw.06} or
\cite{goutte.oudjane.russo.11}. 

We use the scheme \eqref{schemeaimp.eq} and the corresponding
implicit-explicit version of the scheme \eqref{schemeb.eq} to obtain a numerical
approximation of the functions $a$ and $b$. 
The trend is $\mu=0.01 $, and the $C$, $G$, $M$ parameters of the CGMY model are
$C=0.01$, $G=1.1$, and $M=1.1$. The mean reverting coefficient $c$ is equal to $0.1$.
In all experiments, the resolution domain is  $[-10,10]$ and the
integration domain for the Levy density is $[-2,2]$ so that $I=N/5$
and we take $\kappa=1$.

Table \ref{tab: CGMY_a_b_convergence} gives for an at-the-money call
option  and for a number of time steps equal to $800$ the calculated
values  of the $a$ and $b$  with $Y=1.2$, $Y=1.9$ and $Y=1.98$, depending on
$N$ (the parameter $Y$ of the CGMY model corresponds to the $\alpha$
of our main assumptions and must belong to the interval $(1,2)$). We also compute the space discretization error, taking the
approximate value with $N = 3200$ as the reference value. The orders
of convergence $k_a$ and $k_b$ for $a$ and $b$ are computed where possible using the
formula $k_N = \frac{log(e_{N/2}/e_{N })}{log(2)}$, where $e_{N/2}$ and
$e_N$ are the error values for $N/2$ and $N$ discretization steps
respectively. 

We find that for $Y=1.9$ and $Y=1.98$ the order of
convergence as $N\to \infty$ for the functions $a$ and $b$ is somewhat better
than the one predicted by our theoretical result (theorems
\ref{maina.thm} and \ref{mainb.thm}). For $Y=1.2$ the convergence of
the function $a$ is very fast but the
convergence of the function $b$ exhibits oscillations, which make it
difficult to estimate the order of convergence. To make this clear, we
have shown the algebraic values of the error in this case, and the
order of convergence is not computed. 

Table \ref{tab: CGMY_a_b_convergence_time} displays the same
convergence results for time discretization as function of $N_T$,
taking  $N=800$ and using the approximation with $N_T=6400$ as the
reference value.

For all three values of the parameter $Y$, and for both functions, we find first-order
convergence in time. For the function $a$ this is perfectly in line with the
theoretical result (theorem \ref{maina.thm}), but for the function
$b$, the theoretical result (theorem \ref{mainb.thm}) predicts slower
convergence, and can probably be improved.

\begin{table}
\caption{Space discretization convergence for $CGMY$ model (the
  function $b$ is
  computed for an at-the-money call option)} 

\begin{center}

\begin{tabular}{|c|c|c|c|c|c|c|}
\hline
   &  \multicolumn{6}{c|}{$Y=1.2$}  \\
\hline
N      & 100         & 200     &   400       &  800     &   1600 &  3200    \\
$a$  value & 0.85175 & 0.84544    & 0.84233  & 0.84148 & 0.84144 & 0.84143 \\
error      & 0.0103  & 0.0040  &  0.00090 & 0.00005 &   0.00001  &  $-$   \\
$k_a$ & $-$& 1.36 & 2.15 & 4.16 & $-$ &  $-$\\\hline
$b$ value  & 4.8889 & 4.9196   & 4.8990 & 4.89848  & 4.9021  &   4.9029     \\
error      & -0.0140 & 0.0167  & -0.0038  & 0.0044 &  0.0008        &  $-$\\
\hline
\end{tabular}

\medskip

\begin{tabular}{|c|c|c|c|c|c|}
\hline
   &  \multicolumn{5}{c|}{$Y=1.9$}  \\
\hline
N      & 100         & 200     &   400       &  800     &  3200    \\
$a$  value & 0.82495    & 0.82417  & 0.82405  & 0.82405 & 0.82405\\
error & 0.0009 & 0.00012 & 0 & 0 & $-$\\
$k_a$ & $-$& 2.91&$-$&$-$&$-$\\\hline
$b$ value  & 19.2442   & 19.3131 & 19.3279 &  19.3305 &  19.3303    \\
error & 0.0861 & 0.0172 & 0.0024 & 0.0002 & $-$\\
$k_b$    &       $-$    & 2.31    & 2.74    &   3.7    &       $-$          \\
\hline
\end{tabular}

\medskip

\begin{tabular}{|c|c|c|c|c|c|}
\hline
 & \multicolumn{5}{c|}{ $Y=1.98$}   \\
\hline
N      & 100  & 200 & 400 & 800 &  3200   \\
$a$  value &  0.54068     & 0.53974    & 0.53951    &  0.53946   &
0.53946    \\
error &0.00122& 0.00028& 0.00005& 0 & $-$\\
$k_a$        &       $-$       & 2.12       &  2.48      &  $-$          &        $-$      \\\hline
$b$ value  &  41.4783 &  41.5399 &  41.5553  & 41.5596 & 41.5613 \\
error &0.083& 0.0214&0.006& 0.0017&$-$\\
$k_b$    &    $-$       &  1.9     &  1.66     &  1.33   &     $-$    \\
\hline
\end{tabular}
\end{center}

\label{tab: CGMY_a_b_convergence}
\end{table}

\begin{table}
\caption{Time discretization convergence for $CGMY$ model ( the
  function $b$ is computed for an at-the-money call option)} 
\begin{center}
\begin{tabular}{|c|c|c|c|c|c|}
\hline
   &  \multicolumn{5}{c|}{$Y=1.2$}  \\
\hline
$N_T$     & 100         & 200     &   400       &  800   &   6400
\\
$a$  value &  0.84390  & 0.84252 &  0.84182 &  0.84147 &  0.84117  \\
error & 0.00273  &  0.00135  &  0.00065 &   0.0003  &$-$\\
$k_a$ &$-$& 1.016 & 1.055& 1.115 &$-$\\ \hline
$b$ value  & 4.93455   & 4.91409  & 4.90371  & 4.89848  &  4.89388  \\
error &0.04067  &  0.02021 &   0.00983 &   0.0046 &$-$\\
$k_b$    &   $-$     &     1.009     &  1.04  &  1.095   &       $-$         \\
\hline
\end{tabular}

\medskip
\begin{tabular}{|c|c|c|c|c|c|}
\hline
   &  \multicolumn{5}{c|}{$Y=1.9$}  \\
\hline
$N_T$     & 100         & 200     &   400       &  800     &  6400
\\
$a$  value &  0.82449  & 0.82423 & 0.82410  & 0.82404  &  0.82399   \\
error & 0.0005& 0.00024& 0.00011& 0.00005&$-$\\
$k_a$ &$-$& 1.06& 1.13& 1.14 &$-$\\ \hline
$b$ value  &  19.3001 & 19.3176 & 19.3263  &  19.3305 &  19.3342   \\
error &0.0341&0.0166&0.0079&0.0037&$-$\\
$k_b$    &   $-$     &  1.038        &   1.07   &   1.09   &       $-$         \\
\hline
\end{tabular}

\medskip

\begin{tabular}{|c|c|c|c|c|c|}
\hline
 & \multicolumn{5}{c|}{ $Y=1.98$}   \\
\hline
$N_T$      & 100  & 200 & 400 & 800 &  6400   \\
$a$  value &  0.5381   & 0.53892   & 0.53928  & 0.53947    &  0.53963
\\
error &0.00153&0.00071&0.00035&0.00016& $-$\\
$k_a$    &     $-$      &  1.10      &  1.02      &  1.12   &    $-$       \\\hline
$b$ value  &  41.288   & 41.4451   & 41.5217   & 41.5596 & 41.5925\\
error &0.3045&0.1474& 0.0708&0.0329&$-$\\
$k_b$    &       $-$     &  1.04     & 1.05     &   1.10         &     $-$    \\
\hline
\end{tabular}
\end{center}
\label{tab: CGMY_a_b_convergence_time}
\end{table}


Practitioners usually price options of this type and calculate the
hedging strategy assuming that the underlying process $F$ is a
martingale. It is therefore interesting to evaluate the loss of efficiency when using
the hedging strategy computed in the martingale model.
Assuming that $F$ is a martingale means that we should have
$$
F_{d,T,t}:=\frac{1}{d}\int_T^{T+d}\psi(0,s) \exp\left( M(s,t) +l(s)A_t\right)ds
$$
for some $M$ that makes $F$ a martingale under the historical probability $\P$. Using Lemma 15.1 in \cite{cont.tank.04} we obtain
\begin{align*}
M(t,s) = & - \int_0^t \left(\mu+ C \Gamma(1-Y) (M^{Y-1}- G^{Y-1})\right) e^{-c(s-r)} dr \\
 &  +  \int_0^t C \Gamma(-Y) ((M- e^{-c(s-t)})^Y -M^Y + (G+e^{-c(s-t)})^Y -G^Y ) dr 
\end{align*}
First remark that when the underlying process $F$ is a
martingale, $a\equiv 1$: indeed, from PIDE \eqref{QL_PIDE_a}, we have
\begin{align*}
&0 = -\frac{\partial a}{\partial t}-\mu\frac{\partial a}{\partial z}-\int_\R \left (a(t,z+\gamma)-a(t,z)-\gamma\frac{\partial a}{\partial z}(t,z)\right)\nu(dy) 
 - \inf_{|\pi|\leq \bar\Pi}\left\{2\pi\genop[Q]{a}(z) +\pi^2\genop[G]{a} \right\}\\
&a(T,z)=1
\end{align*}
On the other hand, from Definition \ref{operators}, we have
\begin{align*}
\genop[Q]{a}(z) :=& \int_\R \left(e^{\gamma}-1\right)\left(a(t,z+\gamma(t,z,y))-a(t,z)\right)\nu(dy)
\end{align*}
since $\tilde\mu$, given in \eqref{tilde_mu}, is equal to zero (it is
the drift of the process $F$ which is now a martingale). From this, it
is straightforward to deduce that the function $a=1$ is the unique
solution of PIDE \eqref{QL_PIDE_a}. So that, when $F$ is a martingale,
one only needs to compute the function $b$.

We now evaluate the loss of efficiency when using the martingale
hedging strategies compared to the quadratic hedging strategies under
the true historical measure. Our
efficiency comparison criterion is the following: if $H$ is the (put) option and $\theta^{true},\theta^{mart}$ are, respectively, the
optimal quadratic hedging strategy and the martingale strategy, then the efficiency is measured in terms of the standard deviation of the hedged portfolios:
\begin{align}
\text{efficiency}(\theta^{true})^2:=\text{Var} \left(H(F_{d,T,t})-x^{true}-\int_t^T \theta_{r-}^{true}dF_{d,T,r} \right)\label{effcy_1}
\end{align}
where $x^{true}$ is the true optimal price given in \eqref{qh_price}. Similarly
\begin{align}
\text{efficiency}(\theta^{mart})^2:=\text{Var} \left(H(F_{d,T,t})-x^{mart}-\int_t^T \theta_{r-}^{mart}dF_{d,T,r} \right)\label{effcy_2}
\end{align}
where $x^{mart}$ is the price given in \eqref{qh_price} when one uses
the functions $a$ and $b$ computed in the martingale model.,
i.e. $x^{mart}$ is the risk neutral price of $H$. 
The variances are computed by Monte Carlo over 100000 paths using the
rejection method algorithm described in \cite{madan_yor}, with $800$
rebalancing dates in each path. The number of Monte Carlo trajectories
used is limited due to the cost of the simulation algorithm. The trajectories of $F_{d,T,t}$ are simulated using the true model in both cases.
Table \ref{QuadHedgeRes_cgmy} summarizes the results of simulations
with a number of time steps and a number of space discretization steps
equal to $800$ for $t=0$.
\begin{table}
 \begin{center}
\begin{tabular}{|l|l|c|c|c|c|c|}\hline
$Y$  & Option $H$ & Moneyness &   Option value&  $\text{efficiency}(\theta^{true})$      & $\text{efficiency}(\theta^{mart})$ & Variance \\ 
&  &  &  &   &Reduction \\
\hline
 1.2  &  Call    &   1             &  4.67   &     4.93                   &   5.22            &   -5.6   \%        \\ \hline
 1.2  &  Put     &   1             &  4.65   &     4.93                   &    5.22           &   -5.6       \%            \\ \hline
 1.2  &  Call    &   1.1           &  1.03   &     6.33                   &     6.60          &    -4.2   \%                            \\ \hline
 1.2  &  Put     &   1.1           &  8.73   &     6.33                   &     6.60          &   -4.2     \%                    \\ \hline
 1.2  &  Call    &   0.9           &  10.52  &     3.58                   &      3.81         &   -6.4 \%      \\ \hline
 1.2  &  Put     &   0.9           &  2.32   &     3.61                   &     3.82          &  -5.8            \%   \\ \hline
 1.98 &  Call    &    1            &  41.55  &    2.22                    &   3.048           &   -27.6   \%        \\ \hline
 1.98  &  Put     &   1            &  41.63  &    2.19                    &   3.03            &   -27.7    \%            \\ \hline
 1.98 &  Call    &    1.1          &  39.91  &    2.44                    &   3.335           &   -26.7   \%        \\ \hline
 1.98  &  Put     &   1.1          &  47.71  &    2.41                    &   3.32            &    -27.4      \%            \\ \hline
 1.98 &  Call     &   0.9          &  43.35  &    1.98                    &   2.737           &   -27.6   \%        \\ \hline
 1.98  &  Put     &   0.9          &  35.71  &    1.95                    &    2.72           &  -28.3        \%            \\ \hline
\end{tabular}
\end{center}
\caption{Pricing and standard deviation of hedged portfolio in  the CGMY case}
\label{QuadHedgeRes_cgmy}
\end{table}
The numerical experiment proves that one loses efficiency when using
the martingale hedging strategy. This is consistent with the fact that
$\theta^{true}$ achieves the minimum in problem \eqref{elec_prob_2} and outperforms the strategy $\theta^{mart}$.

\subsection{Numerical example: the Normal Inverse Gaussian process}\label{subsec: NIG}
In this last paragraph we study the problem \eqref{elec_prob_2} when $\hat L$ in \eqref{levy_jump} is a Normal Inverse Gaussian process\index{Normal Inverse Gaussian process} with parameters $ \alpha, \beta ,\delta,u $: $\hat L_t  \sim NIG( \alpha, \beta ,\delta  t  , u t )$. 
\begin{remark}
\label{not: NIG}
The parameter $\alpha$ should not be mistaken for the parameter in Lemma \ref{lem_transf}. We use this notation because it is standard in the literature.
\end{remark}
\bigskip
 We can write then
$$
\hat L_t= \left(u + \frac{\beta\delta}{\sqrt{\alpha^2-\beta^2}} \right)t + \int_0^t \int_\R y \tilde J(dyds)
$$
where $\tilde J$ is a compensated Poisson random measure with
intensity 
$$\nu(dy) = \nu(y) dy,\quad \nu(y)= \frac{\alpha\delta }{\pi |y|}K_1(\alpha|y|)
e^{\beta y} ,
$$ 
where $K_1$ is the modified Bessel function of the second kind
(Section 4.4.3 in \cite{cont.tank.04}). The L\'evy density $\nu(y)$ satisfies
\begin{align*}
\nu(y)\stackrel{y\to 0}{\sim}\frac{\delta}{\pi |y|^2},\qquad
\nu(y)\stackrel{y\to +\infty}{\sim}
\frac{1}{|y|^{3/2}}e^{-(\alpha-\beta)y}, \qquad \nu(y)\stackrel{y\to -\infty}{\sim} \frac{1}{|y|^{3/2}}e^{-(\alpha+\beta)|y|}.
\end{align*}

\begin{remark}
\label{rem: no_NIG}
The NIG is a infinite variation  \lev process with stable-like
behavior of small jumps, and since the Blumenthal-Getoor index is
equal to $1$, we cannot formally apply Lemma \ref{lem_transf} and
Theorems \ref{HJB_v0}--\ref{HJB_v}. It is nevertheless a case of
interest because the NIG model is popular among practitioners, and we
shall see in the sequel that our numerical schemes yield acceptable results
for this model. 
\end{remark}
We want to solve numerically the equations
 \eqref{QL_PIDE_a} and \eqref{PIDE_b} where the model for $Z$ is given in
 Lemma \ref{lem_transf} and the source of randomness $\hat L$ in
 \eqref{levy_jump} is given by the NIG process introduced above. We
 apply the scheme \eqref{schemeaimp.eq} for the function $a$ with
 maturity $T = 7$. Once again, the coefficients 
\begin{align*}
\mu(t,z) = & \Phi^{'}( \Phi^{-1}(z))  \left(u + \frac{\beta\delta}{\sqrt{\alpha^2-\beta^2}}\right) e^{c t} 
+ \int_{\R} \left[\Phi\left( \Phi^{-1}(z)+ ye^{ct}\right) - z -  ye^{c t} \Phi^{'}( \Phi^{-1}(z))\right] \nu(dy) \\
\gamma(t,z,y)  = & \Phi\left( \Phi^{-1}(z)+ ye^{c t}\right) - z 
\end{align*}
and the function $D(t,z)$ introduced in \eqref{Dvalue} are computed by
numerical integration over the points $y_i$ such that $\gamma(t,z,y_i(t,z))=i\Delta z$, or equivalently
$$
y_i(t,z):=e^{-c t} \left(\Phi^{-1}(z+i\Delta z)-\Phi^{-1}(z) \right)
$$
By expanding $\gamma$ around zero we obtain 
\begin{align*}
D(t,z):=& \int_{y_{-\kappa-1/2}(t,z)}^{y_{\kappa+1/2}(t,z)}\gamma(t,z,y)^2 \nu(dy) \simeq e^{2c t}(\Phi^{'}(\Phi^{-1}(z)))^2 \int_{y_{-\kappa-1/2}(t,z)}^{y_{\kappa+1/2}(t,z)} y^2 \nu(dy) \\
\simeq & e^{c t}\left(\Phi^{-1}(z+ (\kappa+\frac{1}{2})\Delta z)-\Phi^{-1}(z-(\kappa+\frac{1}{2}) \Delta z)\right) (\Phi^{'}(\Phi^{-1}(z)))^2  \frac{\delta}{\pi}
\end{align*} 
since, around zero, we have $y^2 \nu(dy)  \simeq \frac{\delta}{\pi}  +
\frac{\delta \beta }{\pi} y + O(y^2)$. (See for example
\cite{raible2000}). We proceed as for the CGMY case, even though we do not
have a priori results on the existence of a smooth solution.

We consider once again the problem \eqref{elec_prob_2} for European options $f$, with maturity one week and delivery for the 7 days of the week as in \eqref{future_7}. The parameters of the NIG process are  $u=0.08$, $\alpha =6.23$,
$\beta =0.06$, $\delta=0.1027$. The mean reverting coefficient $c$ is
taken equal to $0.19$. 
In all experiments, the resolution domain is  $[-10,10]$ and the integration domain for the Levy density is $[-2,2]$ so that $I=N/5$ and $\kappa=1$.

Table \ref{tab: NIG_a_b_convergence} displays the values of $a$ and
$b$ for an at-the-money call option as function of the space mesh size
$N$ for the number of time steps $N_T = 800$ and as function of the
number of time steps for the space mesh size $N=800$. Errors and
orders of convergence $k_a$  for $a$ and  $k_b$  for  $b$ are computed
taking $N=3200$ and $N_T=3200$ as the reference values. 

In this case, there are no theoretical results to which the
simulations may be compared. Numerically we do observe convergence in time and
in space, but it seems that the convergence as $N\to \infty$ is slower
than for the CGMY model, in particular, for the function $a$ the space
discretization error seems to be much higher than the time
discretization error. 

\begin{table}
\caption{Space and time discretization convergence for $NIG$ model
  (The value $b$ is computed for an at-the-money call option)} 
\begin{center}
\begin{tabular}{|c|c|c|c|c|c|c|}
\hline
   &  \multicolumn{6}{c|}{ Space}  \\
\hline
$N$     & 100         & 200     &   400       &  800     &  1600  & 3200    \\
$a$  value &  0.398 & 0.266  & 0.182   &  0.136 &  0.113   &  0.1012
\\
error &0.2968   & 0.1648  &  0.0808  &  0.0348  &  0.0118 &$-$  \\
$k_a$    &  $-$  &  0.85   &  1.03  &  1.21       &  1.56 &  $-$ \\
$b$ value  &  4.716 & 4.608  &  4.446  &  4.324 &  4.247    &  4.204
\\
error &0.512  &  0.404   & 0.242   & 0.12   & 0.043& $-$\\
$k_b$    &     $-$   &  0.34        &  0.74  &  1.001   & 1.48  &       $-$       \\
\hline
\end{tabular}

\medskip

\begin{tabular}{|c|c|c|c|c|c|c|}
\hline
 & \multicolumn{6}{c|}{ Time}   \\
\hline
$N_T$      & 100  & 200 & 400 & 800                  &   1600  &  3200   \\
$a$  value &  0.136  & 0.136   & 0.136   &  0.136      &  0.136   &    0.136      \\
$b$ value  &  4.475  & 4.388   &  4.345  &  4.324     &  4.313    &
4.308      \\
error & 0.167  &  0.08  &  0.037  &  0.016   & 0.005 &$-$\\
$k_b$    &   $-$  &  1.06 & 1.11   &   1.2 &  $-$    &   \\
\hline
\end{tabular}
\end{center}
\label{tab: NIG_a_b_convergence}
\end{table}


As in the CGMY case, we estimate the loss in the efficiency of the hedge when using a martingale model in terms of the standard deviation of the hedged portfolios as in \eqref{effcy_1}--\eqref{effcy_2}.
For $F$ to be  a martingale, we should have
$$
F_{d,T,t}:=\frac{1}{d}\int_T^{T+d}\psi(0,s) \exp\left( M(s,t) +l(s)A_t\right)ds
$$
where 
\begin{align*}
&M(t,s) = - \int_0^t \left( \left(u+\frac{\beta\delta}{\sqrt{\alpha^2-\beta^2}}\right) e^{-c(s-r)} + \delta \left( \sqrt{\alpha^2-\beta^2} - \sqrt{\alpha^2 - (\beta + e^{-c(s-r)})^2}\right)\right) dr 
\end{align*}
We already know that in this case we only need to compute the function $b$ ($a=1$ in the martingale case). Table \ref{QuadHedgeRes} summarizes the results of simulations, for $t=0$.
\begin{table}
 \begin{center}
\begin{tabular}{|l|l|c|c|c|c|}\hline
Option $H$ & Moneyness &   Option value&  $\text{efficiency}(\theta^{true})$      & $\text{efficiency}(\theta^{mart})$ & Variance \\ 
&  &  &  &   &Reduction \\
\hline
Call    &   1             &    4.247   &     1.084                      &         1.343     &  -19,28\%              \\ \hline
Put     &   1             &    4.230    &     1.084                     &         1.343     &  -19,28\%    \\ \hline
Call    &   1.5          &     0.120         &     0.142                          &   0.171         &      -16,9           \%       \\ \hline
Put     &   1.5          &     38.809         &    0.145                           &   0.171          &   -15,2             \%   \\ \hline
\end{tabular}
\end{center}
\caption{Pricing and standard deviation of hedged portfolio in the NIG ( $N=1600$, $N_T=800$)
case}
\label{QuadHedgeRes}
\end{table}
The numerical experiment proves that using
the martingale hedging strategy is inaccurate. The loss of efficiency is of order -20\%.

\bibliographystyle{chicago}

\appendix
\section{Convergence analysis}

\begin{proof}[Proof of Theorem \ref{maina.thm}]
Under the additional assumptions of this theorem,
$$
D(t,z) = \int_{\Omega_0}\gamma^2(t,z,y)\nu(dy) \geq
m_1\int_{\Omega_0}y^2\nu(dy) \geq
m_1 \int_{|y|\leq \frac{(\kappa+1/2)\Delta z}{\|\gamma_y\|_\infty}}
y^2\nu(dy) \sim C(\kappa \Delta z)^{2-\alpha},\quad \Delta z\to 0.
$$
where, throughout this appendix, $C$ denotes a finite constant, which
does not depend on any truncation / discretization parameters and
whose exact definition may change from line to line. On the other hand,
$$
|\hat \mu(t,z)| \leq \bar \mu + \Delta z\sum_{k<|i|\leq I} |i|\omega_i(t,z).
$$
Using the mean value theorem and our assumptions on $\gamma$, we can
show that 
$$
|\hat \mu(t,z)| \leq \bar \mu + C \int_{\Omega_1\cup
  \Omega_2}|y|\nu(dy) \leq \bar \mu + C |\kappa \Delta z|^{1-\alpha}.
$$
Therefore, one may choose a constant $c$ which does not depend on
$\Delta z$ (for $\Delta z$ small enough), such that for $\kappa >c$,
the weights $\upsilon$ and $\chi$ are both positive and given by \eqref{centra}--\eqref{centrb}.
Throughout this proof we shall assume that such a
choice of $\kappa$ has been made.

Let 
\begin{align}\hat \beta^{z,t_n} = \inf\{i\geq n: \widehat
Z^{z,t_n}_{t_i} + I\Delta z \geq  N\Delta z\ \text{or}\ \widehat
Z^{z,t_n}_{t_i} - I\Delta z \leq -N\Delta z\}.\label{defbetahat}
\end{align}
We have,
\begin{align*}
a^n(z_j) &= \inf_{\pi_i \in[-\bar \Pi, \bar \Pi]}\mathbb E\Bigg[\prod_{i=n+1}^{\hat \beta^{t_n,z_j} \wedge N_T}\left(1+
\pi_i (e^{\widehat Z^{t_n,z_j}_{t_i}-\widehat Z^{t_n,z_j}_{t_{i-1}}} -1 )\right)^2\\
&\times \prod_{i=\hat \beta^{t_n,z_j} \wedge N_T+1}^{\beta^{t_n,z_j} \wedge N_T}\left(1+
\pi_i (e^{\widehat Z^{t_n,z_j}_{t_i}-\widehat Z^{t_n,z_j}_{t_{i-1}}} -1 )\right)^2q^a(t_{\beta^{t_n,z_j}\wedge
N_T},\widehat Z^{t_n,z_j}_{t_{\beta^{t_n,z_j} \wedge N_T}})\Bigg] \\
& = \inf_{\pi_i \in[-\bar \Pi, \bar \Pi]}\mathbb E\Bigg[\mathbf 1_{\hat \beta^{t_n,z_j} <
  N_T}\prod_{i=n+1}^{ \beta^{t_n,z_j}\wedge N_T }\left(1+
\pi_i (e^{\widehat Z^{t_n,z_j}_{t_i}-\widehat Z^{t_n,z_j}_{t_{i-1}}} -1 )\right)^2 q^a(\beta^{t_n,z_j}\wedge
N_T,\widehat Z^{t_n,z_j}_{\beta^{t_n,z_j} \wedge N_T}) \\
 &+ \mathbf 1_{\hat \beta^{t_n,z_j} \geq N_T}\prod_{i=n+1}^{N_T }\left(1+
\pi_i (e^{\widehat Z^{t_n,z_j}_{t_i}-\widehat Z^{t_n,z_j}_{t_{i-1}}} -1 )\right)^2 a(t_{N_T},\widehat Z^{t_n,z_j}_{ t_{N_T}})\Bigg]
\end{align*}
Let 
$$
\hat a^n(z_j) = \inf_{\pi_i \in[-\bar \Pi, \bar \Pi]} \mathbb E\left[\mathbf 1_{\hat \beta^{t_n,z_j} > N_T-1}\prod_{i=n+1}^{N_T }\left(1+
\pi_i (e^{\widehat Z^{t_n,z_j}_{t_i}-\widehat Z^{t_n,z_j}_{t_{i-1}}} -1 )\right)^2 a(t_{N_T},\widehat Z^{t_n,z_j}_{ t_{N_T}})\right].
$$
Clearly, $a_n(z_j) \geq \hat a_n(z_j)$. On the other hand, by the
Cauchy-Schwarz inequality,
\begin{align*}
a_n(z_j) -\hat a_n(z_j) &\leq \sup_{\pi_i \in[-\bar \Pi, \bar \Pi]}\mathbb E\Bigg[\mathbf 1_{\hat \beta^{t_n,z_j} <
  N_T}\prod_{i=n+1}^{ N_T }\left(1+
\pi_i (e^{\widehat Z^{t_n,z_j}_{t_i}-\widehat Z^{t_n,z_j}_{t_{i-1}}} -1 )\right)^2
\Bigg] \|q^a\|_\infty. \\
&\leq \|q^a\|_\infty \mathbb P[\hat \beta^{t_n,z_j} <
  N_T]^{1\over2} \sup_{\pi_i \in[-\bar \Pi, \bar \Pi]}\mathbb E\Bigg[\prod_{i=n+1}^{ N_T }\left(1+
\pi_i (e^{\widehat Z^{t_n,z_j}_{t_i}-\widehat Z^{t_n,z_j}_{t_{i-1}}} -1 )\right)^4
\Bigg]^{1\over2}
\end{align*}
The second factor can be bounded using the fact, that by Lemma
\ref{cons.lm}, for every $\pi\in [-\bar \Pi,\bar \Pi]]$
\begin{align*}
&\mathbb E[(1+\pi (e^{\widehat Z^{z,t_n}_{t_{n+1}}-z}-1))^4] = 1+ \Delta
t\Big[4\pi\Big\{\mu(t_n,z) + \int_{\mathbb R} (e^{\gamma(t_n,z,y)}-1 -
\gamma(t_n,z,y)) \nu(dz)\Big\}\\
& + 6\pi^2 \int_{\mathbb R} (e^{\gamma(t_n,z,y)}-1)^2 \nu(dz) + 4\pi^3
\int_{\mathbb R} (e^{\gamma(t_n,z,y)}-1)^3 \nu(dz) + \pi^4
\int_{\mathbb R} (e^{\gamma(t_n,z,y)}-1)^4 \nu(dz) + \mathcal E
\Big],
\end{align*}
where $\mathcal E$ is the error term. By our assumptions, the terms in
square brackets are bounded, and therefore, by applying iterated
conditional expectations, 
\begin{align}
\sup_{\pi_i \in[-\bar \Pi, \bar \Pi]}\mathbb E\Bigg[\prod_{i=n+1}^{ N_T }\left(1+
\pi_i (e^{\widehat Z^{t_n,z_j}_{t_i}-\widehat Z^{t_n,z_j}_{t_{i-1}}} -1 )\right)^4
\Bigg] \leq (1+C\Delta t)^{N_T - n}\leq e^{C \Delta T N_T} = e^{CT}. \label{expestim}
\end{align}
Together with the estimate of Lemma \ref{exit.lm} for the first factor, this yields
\begin{align}
|a_n(z_j) -\hat a_n(z_j)| \leq \frac{C(1+|z_j|)}{(N - I)\Delta z - \bar
  \mu T},\quad N\Delta z > I\Delta z + \bar
  \mu T.\label{ahatbound}
\end{align}

The remaining error is decomposed as follows:
\begin{align}
&\hat a^n(z_j) - a(t_n,z_j)  \notag\\&= \sum_{k=n+1}^{N_T} \Big\{\inf_{\pi_i
  \in[-\bar \Pi, \bar \Pi],i=n+1,\dots,k} \mathbb E\left[\mathbf 1_{\hat \beta^{t_n,z_j} > k-1}\prod_{i=n+1}^{k } \left(1+
\pi_i (e^{\widehat Z^{t_n,z_j}_{t_i}-\widehat Z^{t_n,z_j}_{t_{i-1}}}
-1 )\right)^2 a(t_k,\widehat Z^{t_n,z_j}_{t_k})\right] \notag\\ &-
\inf_{\pi_i \in[-\bar \Pi, \bar \Pi],i=n+1,\dots,k-1} \mathbb
E\left[\mathbf 1_{\hat \beta^{t_n,z_j} > k-1}\prod_{i=n+1}^{k-1 } \left(1+
\pi_i (e^{\widehat Z^{t_n,z_j}_{t_i}-\widehat Z^{t_n,z_j}_{t_{i-1}}} -1 )\right)^2 a(t_{k-1},\widehat Z^{t_n,z_j}_{t_{k-1}})\right]\Big\}\label{decomperr}
\end{align}
The first term inside the brackets satisfies
\begin{align*}
&\inf_{\pi_i
  \in[-\bar \Pi, \bar \Pi],i=n+1,\dots,k} \mathbb E\left[\mathbf 1_{\hat \beta^{t_n,z_j} > k-1}\prod_{i=n+1}^{k } \left(1+
\pi_i (e^{\widehat Z^{t_n,z_j}_{t_i}-\widehat Z^{t_n,z_j}_{t_{i-1}}}
-1 )\right)^2 a(t_k,\widehat Z^{t_n,z_j}_{t_k})\right] \\
& = \inf_{\pi_i
  \in[-\bar \Pi, \bar \Pi],i=n+1,\dots,k-1} \mathbb E\Bigg[\mathbf 1_{\hat \beta^{t_n,z_j} > k-1}\prod_{i=n+1}^{k-1 } \left(1+
\pi_i (e^{\widehat Z^{t_n,z_j}_{t_i}-\widehat Z^{t_n,z_j}_{t_{i-1}}}
-1 )\right)^2 \\ &\qquad \qquad \times \inf_{\pi_k \in [-\bar \Pi,\bar
\Pi]}\mathbb E\left[\left(1+
\pi_k (e^{\widehat Z^{t_n,z_j}_{t_k}-\widehat Z^{t_n,z_j}_{t_{k-1}}}
-1 )\right)^2 a(t_k,\widehat Z^{t_n,z_j}_{t_k})\Big|
\widehat Z^{t_n,z_j}_{t_{k-1}}\right]\Bigg]
\end{align*}
By Lemma \ref{cons.lm}, for every $\pi\in [-\bar \Pi, \bar \Pi]$, 
\begin{align*}
&\mathbb E\left[\left(1+
\pi (e^{\widehat Z^{t_n,z_j}_{t_k}-\widehat Z^{t_n,z_j}_{t_{k-1}}}
-1 )\right)^2 a(t_k,\widehat Z^{t_n,z_j}_{t_k})\Big|
\widehat Z^{t_n,z_j}_{t_{k-1}}=z\right] \\
&= \mathbb E\left[\left(1+
\pi (e^{\widehat Z^{t_n,z_j}_{t_k}-\widehat Z^{t_n,z_j}_{t_{k-1}}}
-1 )\right)^2 a(t_{k-1},\widehat Z^{t_n,z_j}_{t_k})\Big|
\widehat Z^{t_n,z_j}_{t_{k-1}}=z\right] \\
& + \int_{t_{k-1}}^{t_k} \mathbb E\left[\left(1+
\pi (e^{\widehat Z^{t_n,z_j}_{t_k}-\widehat Z^{t_n,z_j}_{t_{k-1}}}
-1 )\right)^2 \frac{\partial a}{\partial t}(t,\widehat Z^{t_n,z_j}_{t_k})\Big|
\widehat Z^{t_n,z_j}_{t_{k-1}}=z\right] dt\\
& = a(t_{k-1},z) + \Delta t \mathcal L_\pi a(t_{k-1},z) + \Delta t
\int_{t_{k-1}}^{t_k} \mathcal L_\pi \frac{\partial a}{\partial t}(t,z) dt + a(t_k,z) -
a(t_{k-1},z) + \Delta t \tilde{\mathcal E}_k,
\end{align*}
where 
$$
|\tilde{\mathcal E}_k| \leq C \Delta 
z^{3-\alpha}(\kappa^{3-\alpha} + \kappa^{1-\alpha}) + C \int_{|y|\geq \frac{I
    \Delta z}{\|\gamma'_y\|_\infty}} (1+|y| + \tau(y) + \tau^2(y))\nu(dy),
$$
and the operator $\mathcal L_\pi$ is defined by
\begin{align*}
\mathcal L_\pi f(t,z) &= \mu(t,z)f'(z) + \int_{\mathbb R} (f(z+\gamma(t,z,y))-
  f(z)-f'(z)\gamma(t,z,y))\nu(dy) \\
& + 2\pi \left\{\mu(t,z)f(z) + \int_{\mathbb R} ((e^{\gamma(t,z,y)}-1)f(z+\gamma(t,z,y))
-f(z)\gamma(t,z,y))\nu(dy)\right\}\\
& + \pi^2 \int_{\mathbb R} (e^{\gamma(t,z,y)}-1)^2f(z+\gamma(t,z,y))
\nu(dy).
\end{align*}
From the regularity of $a$ and the integrability conditions on the
L\'evy measure, we deduce that $\mathcal L_\pi \frac{\partial a}{\partial t}(t,z)$ is uniformly
bounded on $\pi \in [-\bar\Pi,\bar \Pi]$, which means that 
\begin{align*}
&\mathbb E\left[\left(1+
\pi (e^{\widehat Z^{t_n,z_j}_{t_k}-\widehat Z^{t_n,z_j}_{t_{k-1}}}
-1 )\right)^2 a(t_k,\widehat Z^{t_n,z_j}_{t_k})\Big|
\widehat Z^{t_n,z_j}_{t_{k-1}}=z\right] \\ &\qquad = a(t_{k-1},z) + \Delta t\left\{
\mathcal L_\pi a(t_{k-1},z) + \frac{\partial a}{\partial t}(t_{k-1},z)\right\}
+ \Delta t \mathcal E_k
\end{align*}
with 
\begin{align}
|\mathcal E_k| \leq C \overline{\mathcal E},\quad \overline{\mathcal E}:=  \Delta t + \Delta 
z^{3-\alpha}(\kappa^{3-\alpha} + \kappa^{1-\alpha}) +  \int_{|y|\geq \frac{I
    \Delta z}{\|\gamma'_y\|_\infty}} (1+|y| + \tau(y) + \tau^2(y))\nu(dy),\label{bounderr.eq}
\end{align}
Using the equation \eqref{QL_PIDE_a} satisfied by $a$, we finally get
$$
\left|\inf_{\pi \in [-\bar \Pi, \bar \Pi]} \mathbb E\left[\left(1+
\pi (e^{\widehat Z^{t_n,z_j}_{t_k}-\widehat Z^{t_n,z_j}_{t_{k-1}}}
-1 )\right)^2 a(t_k,\widehat Z^{t_n,z_j}_{t_k})\Big|
\widehat Z^{t_n,z_j}_{t_{k-1}}=z\right] - a(t_{k-1},z) \right|
\leq C\Delta t \overline{\mathcal E}.
$$
Plugging this estimate back into \eqref{decomperr} yields
\begin{align*}
&\hat a^n(z_j) - a(t_n,z_j)  \notag\\&\leq \sum_{k=n+1}^{N_T} \Big\{\inf_{\pi_i
  \in[-\bar \Pi, \bar \Pi],i=n+1,\dots,k-1} \mathbb E\left[\mathbf 1_{\hat \beta^{t_n,z_j} > k-1}\prod_{i=n+1}^{k-1 } \left(1+
\pi_i (e^{\widehat Z^{t_n,z_j}_{t_i}-\widehat Z^{t_n,z_j}_{t_{i-1}}}
-1 )\right)^2 (a(t_{k-1},\widehat Z^{t_n,z_j}_{t_{k-1}})+ C\Delta t
\overline{\mathcal E})\right] \notag\\ &-
\inf_{\pi_i \in[-\bar \Pi, \bar \Pi],i=n+1,\dots,k-1} \mathbb
E\left[\mathbf 1_{\hat \beta^{t_n,z_j} > k-1}\prod_{i=n+1}^{k-1 } \left(1+
\pi_i (e^{\widehat Z^{t_n,z_j}_{t_i}-\widehat Z^{t_n,z_j}_{t_{i-1}}} -1 )\right)^2 a(t_{k-1},\widehat Z^{t_n,z_j}_{t_{k-1}})\right]\Big\},
\end{align*}
which implies:
$$
\hat a^n(z_j) - a(t_n,z_j)  \leq C\Delta t \overline{\mathcal E} \sum_{k=n+1}^{N_T} \sup_{\pi_i
  \in[-\bar \Pi, \bar \Pi],i=n+1,\dots,k-1} \mathbb E\left[\prod_{i=n+1}^{k-1 } \left(1+
\pi_i (e^{\widehat Z^{t_n,z_j}_{t_i}-\widehat Z^{t_n,z_j}_{t_{i-1}}}
-1 )\right)^2 \right]
$$
The expectation can be estimated as in \eqref{expestim}, and we
finally get
$$
\hat a^n(z_j) - a(t_n,z_j) \leq C \overline{\mathcal E}.
$$
The upper bound can be obtained in a similar manner. 
\end{proof}
\begin{proof}[Proof of Corollary \ref{hedging.cor}]
We first estimate the error of approximating the numerator and the
denominator in \eqref{defpin}. For the numerator, we get:
\begin{align*}
&\left|  \mathbb E\left[\left(e^{Z^{t_n,z_j}_{t_{n+1}}-z_j}-1\right)a_{n+1}(Z^{t_n,z_j}_{t_{n+1}})\right]
  - \Delta t\mathcal Q_t a(t_n,z_j)\right|\\
&\leq \left|
  \mathbb E\left[\left(e^{Z^{t_n,z_j}_{t_{n+1}}-z_j}-1\right)(a_{n+1}(Z^{t_n,z_j}_{t_{n+1}})-a(t_{n+1},Z^{t_n,z_j}_{t_{n+1}}))\right]\right|\\
&+ \left|
  \mathbb E\left[\left(e^{Z^{t_n,z_j}_{t_{n+1}}-z_j}-1\right)(a(t_{n+1},Z^{t_n,z_j}_{t_{n+1}})
    - a(t_n,Z^{t_n,z_j}_{t_{n+1}}))\right]\right|\\
&+\left|  \mathbb E\left[\left(e^{Z^{t_n,z_j}_{t_{n+1}}-z_j}-1\right)a(t_n,Z^{t_n,z_j}_{t_{n+1}})\right]
  - \Delta t\mathcal Q_t a(t_n,z_j)\right|
\end{align*}
Using Lemma \ref{cons.lm}, we can show that the second term is bounded
by $C \Delta t^2$ and the third term is bounded by $C \Delta t
\overline{\mathcal E}$, with $\overline{\mathcal E}$ defined in
\eqref{bounderr.eq}. The first term makes the main contribution to the
error, which can be estimated using Theorem \ref{maina.thm} and the
Cauchy-Schwarz inequality:
$$
\left|
 \mathbb  E\left[\left(e^{Z^{t_n,z_j}_{t_{n+1}}-z_j}-1\right)(a_{n+1}(Z^{t_n,z_j}_{t_{n+1}})-a(t_{n+1},Z^{t_n,z_j}_{t_{n+1}}))\right]\right|\leq
C \overline{\mathcal E}
 \mathbb E\left[\left(e^{Z^{t_n,z_j}_{t_{n+1}}-z_j}-1\right)^2
  \right]^{\frac{1}{2}} \leq C \sqrt{\Delta t} \overline{\mathcal E}
$$
Similarly, for the denominator we have the estimate 
$$
\left|  \mathbb E\left[\left(e^{Z^{t_n,z_j}_{t_{n+1}}-z_j}-1\right)^2 a_{n+1}(Z^{t_n,z_j}_{t_{n+1}})\right]
  - \Delta t\mathcal G_t a(t_n,z_j)\right| \leq C \Delta t \overline{\mathcal
  E}. 
$$
Using the fact that for $b,b'>0$ and all $a,a'$, 
\begin{align*}
  \left|\frac{a'}{b'} - \frac{a}{b}\right|\leq \frac{|a'-a|}{b} +
  \frac{|a'| |b'-b|}{bb'},
\end{align*}
we obtain the estimate
\begin{align*}
&\left|\frac{\mathbb
  E\left[\left(e^{Z^{t_n,z_j}_{t_{n+1}}-z_j}-1\right)a_{n+1}(Z^{t_n,z_j}_{t_{n+1}})\right]}{\mathbb
  E\left[\left(e^{Z^{t_n,z_j}_{t_{n+1}}-z_j}-1\right)^2a_{n+1}(Z^{t_n,z_j}_{t_{n+1}})\right]}
- \frac{\mathcal Q_t a(t_n,z_j)}{\Delta t\mathcal G_t
  a(t_n,z_j)}\right| \\ &\leq
\frac{C\sqrt{\Delta T}\overline{\mathcal
  E}}{\Delta t \mathcal G_t a(t_n,z_j)} +  \frac{C \Delta t^{\frac{3}{2}} \overline{\mathcal
  E}}{\Delta t^2\mathcal G_t a(t_n,z_j) (\mathcal G_t a(t_n,z_j) - C \overline{\mathcal
  E})}\leq C \Delta t^{-\frac{1}{2}}\overline{\mathcal E}
\end{align*}
since $\mathcal G_t a$ is bounded from below (this follows from
Theorem \ref{theo_a_lip}). We conclude by observing that projecting
both the optimal strategy and its approximation on the interval
$[-\bar \Pi, \bar \Pi]$ does not increase the error. 
\end{proof}

\begin{proof}[Proof of Theorem \ref{mainb.thm}]
As in the proof of Theorem \ref{maina.thm}, we may and will assume
that $\upsilon$ and $\chi$ are positive and given by
\eqref{centra}--\eqref{centrb}. To deal with the domain truncation, we
also, similarly to the proof of Theorem \ref{maina.thm}, take $\hat
\beta^{z,t_n}$ as in \eqref{defbetahat}, and use the representation
\begin{align*}
b^n(z_j) &= \mathbb E\Bigg[\mathbf 1_{\hat \beta^{t_n,z_j} <
  N_T}\prod_{i=n+1}^{ \beta^{t_n,z_j}\wedge N_T }\left(1+
\pi^{*i}(\widehat Z^{t_n,z_j}_{t_i}) (e^{\widehat Z^{t_n,z_j}_{t_i}-\widehat Z^{t_n,z_j}_{t_{i-1}}} -1 )\right) q^b(t_{\beta^{t_n,z_j}\wedge
N_T},\widehat Z^{t_n,z_j}_{t_{\beta^{t_n,z_j} \wedge N_T}}) \\
 &+ \mathbf 1_{\hat \beta^{t_n,z_j} \geq N_T}\prod_{i=n+1}^{N_T }\left(1+
\pi^{*i}(\widehat Z^{t_n,z_j}_{t_i}) (e^{\widehat Z^{t_n,z_j}_{t_i}-\widehat Z^{t_n,z_j}_{t_{i-1}}} -1 )\right) f(\widehat Z^{t_n,z_j}_{ t_{N_T}})\Bigg]
\end{align*}
Letting 
$$
\hat b^n(z_j) = \mathbf 1_{\hat \beta^{t_n,z_j} > N_T-1}\prod_{i=n+1}^{N_T }\left(1+
\pi^{*i}(\widehat Z^{t_n,z_j}_{t_i}) (e^{\widehat Z^{t_n,z_j}_{t_i}-\widehat Z^{t_n,z_j}_{t_{i-1}}} -1 )\right) f(\widehat Z^{t_n,z_j}_{ t_{N_T}})\Bigg],
$$
the remainder can be estimated as follows:
\begin{align*}
|b_n(z_j) - \hat b_n(z_j)| &\leq \|q^b\|_\infty \mathbb P[\hat \beta^{t_n,z_j} <
  N_T]^{1\over2} \mathbb E\Bigg[\prod_{i=n+1}^{ N_T }\left(1+
\pi^{*i}(\widehat Z^{t_n,z_j}_{t_i}) (e^{\widehat Z^{t_n,z_j}_{t_i}-\widehat Z^{t_n,z_j}_{t_{i-1}}} -1 )\right)^2
\Bigg]^{1\over2} \\&\leq \|q^b\|_\infty \mathbb P[\hat \beta^{t_n,z_j} <
  N_T]^{1\over2} ,
\end{align*}
where the last inequality uses the fact that $\pi^*$ is the minimizer
of the expectation in the first line, and that the expectation equals
one when substituting the value $\pi^*=0$. Therefore, $|b_n(z_j) -
\hat b_n(z_j)|$ admits the same bound as $|a_n(z_j) - \hat a_n(z_j)|$,
given by \eqref{ahatbound}. 

The remaining error is decomposed as follows:
\begin{align*}
&\hat b^n(z_j) - b(t_n,z_j) = \sum_{k=n+1}^{N_T} \Bigg\{\mathbb E\Bigg[\mathbf 1_{\hat \beta^{t_n,z_j} >k-1}\prod_{i=n+1}^{k-1}\left(1+
\pi^{*i}(\widehat Z^{t_n,z_j}_{t_i}) (e^{\widehat
  Z^{t_n,z_j}_{t_i}-\widehat Z^{t_n,z_j}_{t_{i-1}}} -1 )\right)\\
&\times \mathbb E\left[\left(1+
\pi^{*k}(\widehat Z^{t_n,z_j}_{t_k}) (e^{\widehat
  Z^{t_n,z_j}_{t_k}-\widehat Z^{t_n,z_j}_{t_{k-1}}} -1 )\right)
b(t_k,\widehat Z^{t_n,z_j}_{t_k}) - b(t_{k-1},\widehat
Z^{t_n,z_j}_{t_{k-1}})\Big|\widehat
Z^{t_n,z_j}_{t_{k-1}} \right]\Bigg]\Bigg\}
\end{align*}
Similarly to the proof of Theorem \ref{maina.thm}, by Lemma
\ref{cons.lm}, we deduce that 
\begin{align*}
&\mathbb E\left[\left(1+
\pi^{*k}(z) (e^{\widehat
  Z^{t_n,z_j}_{t_k}-\widehat Z^{t_n,z_j}_{t_{k-1}}} -1 )\right)
b(t_k,\widehat Z^{t_n,z_j}_{t_k}) - b(t_{k-1},\widehat
Z^{t_n,z_j}_{t_{k-1}})\Big|\widehat
Z^{t_n,z_j}_{t_{k-1}} =z\right]\\
& = \Delta t \left\{(\mathcal B_t - \mathcal A_t) b(t_{k-1},z) +
\pi^{*k}(z)\mathcal Q_t b(t_{k-1},z) + \frac{\partial b}{\partial
  t}(t_{k-1},z) \right\} + \Delta t \mathcal E_k\\
& = \Delta t (\pi^{*k}(z)-\pi^{*}(t_{k-1},z))\mathcal Q_t b(t_{k-1},z)  + \Delta t \mathcal E_k
\end{align*}
with $\mathcal E_k$ satisfying \eqref{bounderr.eq}, where in the last
equality we used equation \eqref{PIDE_b}. We conclude by substituting
the estimate for $|\pi^{*k}(z)-\pi^{*}(t_{k-1},z)|$ given in corollary
\eqref{hedging.cor}, and using the fact that $\mathcal Q_t
b(t_{k-1},z)$ is bounded due to the regularity of $b$. 
\end{proof}
\paragraph{Auxiliary lemmas} 
\begin{lemma}\label{cons.lm}
Let $f$ be 4 times continuously differentiable with bounded
derivatives. Then,
\begin{align}
&\left|\frac{\mathbb E[f(\widehat Z^{z,t_n}_{t_{n+1}})] - f(z)}{\Delta
  t} -\mu(t_n,z)f'(z) - \int_{\mathbb R} (f(z+\gamma(t_n,z,y))-
  f(z)-f'(z)\gamma(t_n,z,y))\nu(dy) \right|\notag\\&\leq C \Delta 
z^{3-\alpha}(k^{3-\alpha} + k^{1-\alpha}) + C \int_{|y|\geq \frac{I
    \Delta z}{\|\gamma_y\|_\infty}} (1+|y|)\nu(dy)
\label{cons1}\\
&\left|\frac{\mathbb E[(e^{\widehat Z^{z,t_n}_{t_{n+1}}-z}-1)
    f(\widehat Z^{z,t_n}_{t_{n+1}})] }{\Delta t} -\mu(t_n,z)f(z) - \int_{\mathbb R} ((e^{\gamma(t_n,z,y)}-1)f(z+\gamma(t_n,z,y))
-f(z)\gamma(t_n,z,y))\nu(dy)\right|  \notag\\&\leq C \Delta 
z^{3-\alpha}(k^{3-\alpha} + k^{1-\alpha}) + C \int_{|y|\geq \frac{I
    \Delta z}{\|\gamma_y\|_\infty}} \tau(y)\nu(dy)\label{cons2}\\
&\left|\frac{\mathbb E[(e^{\widehat Z^{z,t_n}_{t_{n+1}}-z}-1)^p
    f(\widehat Z^{z,t_n}_{t_{n+1}})]}{\Delta t}  -\int_{\mathbb R} (e^{\gamma(t_n,z,y)}-1)^pf(z+\gamma(t_n,z,y))
\nu(dy) \right|\notag\\&\leq C \Delta 
z^{3-\alpha}(k^{3-\alpha} + k^{1-\alpha}) + C \int_{|y|\geq \frac{I
    \Delta z}{\|\gamma_y\|_\infty}} \tau^p(y)\nu(dy),\quad p=2,3,4,\label{cons3}
\end{align}
where $\tau$ is defined in Assumption \ref{assumptions_1}-I. 
\end{lemma}
\begin{proof}
We begin with \eqref{cons1}. By definition of $\widehat Z$, the
expression inside the absolute value on the left-hand side may be
decomposed into the following terms:
\begin{align}
&\frac{f(z+\Delta z) + f(z-\Delta z) - 2f(z) - f^{\prime\prime}(z)
  \Delta z^2}{2\Delta z^2}D(t_n,z) \label{line1}\\
&- \int_{ \Omega_0(t_n,z)} \left\{f(z + \gamma(t_n,z,y)) - f(z)
  - f'(z)\gamma(t_n,z,y) - \frac{1}{2} f^{\prime\prime}(z)
  \gamma^2(t_n,z,y)\right\}\nu(dy) \label{line2}\\
& + \frac{f(z+\Delta z) - f(z-\Delta z) - 2f'(z)\Delta z }{2\Delta
  z}\hat \mu(t_n,z) \label{line3}\\
& - \int_{\Omega_3(t_n,z)} (f(z+\gamma(t_n,z,y))-
  f(z) - \gamma(t_n,z,y)f'(z))\nu(dy) \label{line4}\\
& + \sum_{i:y_i(t_n,z) \in \Omega_2}
\int_{y_{i-1/2}(t,z)}^{y_{i+1/2}(t,z) }(F(z,\gamma(t_n,z,y_i)) -
F(z,\gamma(t_n,z,y)))\nu(dy) \label{line5}\\
& + \sum_{i:y_i(t_n,z) \in \Omega_1}
\int_{y_{i-1/2}(t,z)}^{y_{i+1/2}(t,z) }(\widetilde
F(z,\gamma(t_n,z,y_i))-\widetilde F(z,\gamma(t_n,z,y)))y^2\nu(dy), \label{line6}
\end{align}
where we denote
\begin{align*}
&F(z,y) = f(z+y) - f(z) - yf'(z)\\
&\text{and}\quad \widetilde F(z,y) = \frac{f(z+y) - f(z) - yf'(z)}{\gamma^{-1}(t,z,y)^2},
\end{align*}
The function $F$ clearly satisfies
$$
\|F'_y\|_\infty \leq 2 \|f'\|_\infty,\quad \|F^{\prime\prime}_y\|_\infty \leq  \|f^{\prime\prime}\|_\infty.
$$
As for the function $\widetilde F$, we can write it as $\widetilde F =
u(z,y)v^2(t,z,y)$ with 
\begin{align*}
&u(z,y) = \frac{f(z+y) - f(z) - yf'(z)}{y^2} = \int _0^1
f^{\prime\prime}(z + \theta y) (1-\theta)d\theta,\\
& \text{and}\quad
v(t,z,y) = \frac{y}{\gamma^{-1}(t,z,y)} = \int_0^1 \gamma_y(\theta\gamma^{-1}(y))d\theta.
\end{align*}
It is easy to see that under our assumptions the functions $u$ and $v$
are bounded together with their first and second derivatives. It
follows that $\|\widetilde F'_y\|_\infty <\infty$ and $
\|\widetilde F^{\prime\prime}_y\|_\infty <\infty$.

Remark also that by the mean value theorem and our assumptions on $\gamma$,
$$
\left|\sum_{i:y_i(t,z)\in \Omega_1 \cup
  \Omega_2} \omega_i(t,z) \gamma(t,z,y_i(t,z))\right|\leq
C\int_{\Omega_1 \cup \Omega_2} |\gamma(t,z,y)|
\nu(dy) 
$$
for some constant $C$ which does not depend on truncation /
discretization parameters. 

The terms in (\ref{line1}-\ref{line6}) admit the following bounds:
\begin{align*}
|\eqref{line1}| &\leq \frac{\|f^{(4)}\|_\infty}{24} \Delta z^2
\int_{\Omega_0}\gamma^2(t_n,z,y)\nu(dy) \leq
\frac{\|f^{(4)}\|_\infty}{24} \Delta z^2 \|\gamma_y\|^2_\infty
\int_{y \leq \frac{(k+1/2)\Delta z}{m_1}} y^2\nu(dy)\\ & \leq C \Delta z^2
(k \Delta z)^{2-\alpha},
\end{align*}
\begin{align*}
|\eqref{line2}| &\leq \frac{\|f^{(3)}\|}{6} \int_{\Omega_0}
\gamma^3(t_n,z,y)\nu(dy)\leq C(k\Delta z)^{3-\alpha},\\
|\eqref{line3}| &\leq \frac{\|f^{(3)}\|}{6} \Delta z^2 (\bar \mu  +
\int_{\Omega_1 \cup \Omega_2}
|\gamma(t_n,z,y)|\nu(dy)) \\ &\leq \frac{\|f^{(3)}\|}{6} \Delta z^2 (\bar \mu  +\|\gamma_y\|
\int_{ |y|\geq \frac{(k+1/2)\Delta z}{\|\gamma_y\|_\infty}}
|y|\nu(dy))\leq C\Delta z^2 (k\Delta z)^{1-\alpha},\\
|\eqref{line4}| & \leq 2\|f\|_\infty \nu(\Omega_3)
+\|f'\|_\infty \int_{\Omega_3} |\gamma(t,z,y)|\nu(dy)  \leq C
\int_{|y|\geq \frac{I\Delta z}{\|\gamma_y\|_\infty} }(1+|y|)\nu(dy),\\
|\eqref{line5}| & = \left|\sum_{i:y_i(t_n,z) \in \Omega_2}
\int_{(i-1/2)\Delta z}^{(i+1/2)\Delta z }(F(z,\zeta) -
F(z,i\Delta
z))\frac{\nu(\gamma^{-1}(\zeta))}{\gamma_y(\gamma^{-1}(\zeta))}d\zeta\right|\\
&\leq \sum_{i:y_i(t_n,z) \in \Omega_2} \left\{\frac{\Delta z^2}{4}
\|F_{yy}\|_\infty \int_{y_{i-1/2}}^{y_{i+1/2}} \nu(dy) +
\|F_{y}\|_\infty \left|\int_{(i-1/2)\Delta z}^{(i+1/2)\Delta z }(\zeta -i\Delta
z)\frac{\nu(\gamma^{-1}(\zeta))}{\gamma_y(\gamma^{-1}(\zeta))}d\zeta\right|\right\}\\
&\leq C \Delta z^2\int_{\Omega_2} \nu(dy)+ \sum_{i:y_i(t_n,z)
  \in \Omega_2} \|F_{y}\|_\infty\frac{\Delta
  z^2}{4}\int_{(i-1/2)\Delta z}^{(i+1/2)\Delta z
}\left|\frac{d}{d\zeta}
  \frac{\nu(\gamma^{-1}(\zeta))}{\gamma_y(\gamma^{-1}(\zeta))}\right|
d\zeta\\
&\leq C\Delta z^2 \left(\int_{\Omega_2} \nu(dy) +
  \int_{\Omega_2} |\nu'(y)| dy\right)\leq C \Delta z^2. 
\end{align*}
Similarly, 
\begin{align*}
|\eqref{line6}| & = \left|\sum_{i:y_i(t_n,z) \in \Omega_1}
\int_{(i-1/2)\Delta z}^{(i+1/2)\Delta z }(\widetilde F(z,\zeta) -
\widetilde F(z,i\Delta
z)) \gamma^{-1}(\zeta)^2\frac{\nu(\gamma^{-1}(\zeta))}{\gamma_y(\gamma^{-1}(\zeta))}d\zeta\right|\\
&\leq \frac{\Delta
  z^2}{4}\sum_{i:y_i(t_n,z) \in \Omega_1} \left\{
\|\widetilde F_{yy}\|_\infty \int_{y_{i-1/2}}^{y_{i+1/2}}
y^2 \nu(dy) + \|\widetilde F_{y}\|_\infty\int_{(i-1/2)\Delta z}^{(i+1/2)\Delta z
}\left|\frac{d}{d\zeta}\left(
  \gamma^{-1}(\zeta)^2\frac{\nu(\gamma^{-1}(\zeta))}{\gamma_y(\gamma^{-1}(\zeta))}\right)\right|
d\zeta\right\}\\
&\leq C\Delta z^2 \left(\int_{\Omega_1} y^2 \nu(dy) +\int_{\Omega_1} |y| \nu(dy) +
  \int_{\Omega_1} |y^2\nu'(y)| dy\right)\leq C \Delta z^2
(k\Delta z)^{1-\alpha}. 
\end{align*}
We finish the proof of \eqref{cons1} by observing that given that $k\Delta z$ is small
(this is the small jump truncation level), the leading contribution is
made by terms \eqref{line2}, \eqref{line3}, \eqref{line4} and
\eqref{line6}. 

We next prove \eqref{cons2}. The
expression inside the absolute value on the left-hand side may be
decomposed into the following terms:
\begin{align}
&\frac{(e^{\Delta z}-1)f(z+\Delta z) + (e^{-\Delta z}-1)f(z-\Delta z) - (f(z)+2f'(z))
  \Delta z^2}{2\Delta z^2}D(t_n,z) \label{line1.2}\\
&- \int_{\Omega_0(t_n,z)} \left\{(e^{\gamma(t,z,y)}-1)f(z + \gamma(t_n,z,y)) 
  - f(z)\gamma(t_n,z,y) - \frac{1}{2} (f(z)+2f^{\prime}(z))
  \gamma^2(t_n,z,y)\right\}\nu(dy) \label{line2.2}\\
& + \frac{(e^{\Delta z}-1)f(z+\Delta z) - (e^{-\Delta z}-1)f(z-\Delta z) - 2f(z)\Delta z }{2\Delta
  z}\hat \mu(t_n,z) \label{line3.2}\\
& - \int_{\Omega_3(t_n,z)} ((e^{\gamma(t_n,z,y)}-1)f(z+\gamma(t_n,z,y))-
  \gamma(t_n,z,y)f(z))\nu(dy) \label{line4.2}\\
& + \sum_{i:y_i(t_n,z) \in \Omega_2}
\int_{y_{i-1/2}(t,z)}^{y_{i+1/2}(t,z) }(G(z,\gamma(t_n,z,y_i)) -
G(z,\gamma(t_n,z,y)))\nu(dy) \label{line5.2}\\
& + \sum_{i:y_i(t_n,z) \in \Omega_1}
\int_{y_{i-1/2}(t,z)}^{y_{i+1/2}(t,z) }(\widetilde
G(z,\gamma(t_n,z,y_i))-\widetilde G(z,\gamma(t_n,z,y)))y^2\nu(dy), \label{line6.2}
\end{align}
where now the functions
\begin{align*}
G(z,y) := (e^{y}-1)f(z+y) -y f(z)\quad \text{and}\quad \widetilde G(z,y) := \frac{(e^y-1)f(z+y) - yf(z)}{\gamma^{-1}(t,z,y)^2}
\end{align*}
satisfy
$$
|G'_y| + |G^{\prime\prime}_y\| \leq C(e^y + 1),\quad y \in \mathbb R
$$
and 
$$
|\widetilde G'_y| + |\widetilde G^{\prime\prime}_y|  \leq C, \quad |y|\leq 1.
$$
The latter estimate, in particular, follows from the Taylor formula representation
$$
\frac{(e^y-1)f(z+y) - yf(z)}{y^2} = \int_0^1 \{e^{\theta y}(f(z+\theta
y)) + 2 e^{\theta y}(f'(z+\theta
y))  + (e^{\theta y}-1)(f^{\prime\prime}(z+\theta
y)) \}(1-\theta)d\theta.
$$
The different terms can be estimated in a manner, similar to the first
part of the proof:
\begin{align*}
&|\eqref{line1.2} | \leq C \Delta z^2 (k\Delta z)^{2-\alpha},\qquad
|\eqref{line2.2} | \leq C  (k\Delta z)^{3-\alpha},\qquad
|\eqref{line3.2}|\leq C \Delta z^2 (k\Delta z)^{1-\alpha},\\
&|\eqref{line4.2}| \leq C
\int_{|y|\geq \frac{I\Delta z}{\|\gamma_y\|_\infty}
}\tau(y)\nu(dy),\\ &|\eqref{line5.2}| \leq C\Delta z^2 \left(\int_{\Omega_2} (e^{\gamma(t_n,z,y)}+1)\nu(dy) +
  \int_{\Omega_2} (e^{\gamma(t_n,z,y)}+1)|\nu'(y)|
  dy\right)\leq C \Delta z^2\\
&|\eqref{line6.2}| \leq C \Delta z^2 (k\Delta z)^{1-\alpha}.
\end{align*}

Finally, \eqref{cons3} can be proven in a similar manner. For example,
for $p=2$, we
decompose the 
expression inside the absolute value on the left-hand side into the
following terms:
\begin{align*}
&\frac{(e^{\Delta z}-1)^2f(z+\Delta z) + (e^{-\Delta z}-1)^2f(z-\Delta z) - 2f(z)
  \Delta z^2}{2\Delta z^2}D(t_n,z) \\
&- \int_{\Omega_0(t_n,z)} \left\{(e^{\gamma(t,z,y)}-1)^2 f(z + \gamma(t_n,z,y)) 
  -f(z)
  \gamma^2(t_n,z,y)\right\}\nu(dy) \\
& + \frac{(e^{\Delta z}-1)^2f(z+\Delta z) - (e^{-\Delta z}-1)^2f(z-\Delta z) }{2\Delta
  z}\hat \mu(t_n,z) \\
& - \int_{\Omega_3(t_n,z)} (e^{\gamma(t_n,z,y)}-1)^2f(z+\gamma(t_n,z,y))\nu(dy) \\
& + \sum_{i:y_i(t_n,z) \in \Omega_2}
\int_{y_{i-1/2}(t,z)}^{y_{i+1/2}(t,z) }(H(z,\gamma(t_n,z,y_i)) -
H(z,\gamma(t_n,z,y)))\nu(dy) \\
& + \sum_{i:y_i(t_n,z) \in \Omega_1}
\int_{y_{i-1/2}(t,z)}^{y_{i+1/2}(t,z) }(\widetilde
H(z,\gamma(t_n,z,y_i))-\widetilde H(z,\gamma(t_n,z,y)))y^2\nu(dy), 
\end{align*}
with 
\begin{align*}
H(z,y) := (e^{y}-1)^2f(z+y)\quad \text{and}\quad \widetilde H(z,y) := \frac{(e^y-1)^2f(z+y) }{\gamma^{-1}(t,z,y)^2}.
\end{align*}
\end{proof}

\begin{lemma}\label{exit.lm}
There exists a constant $C<\infty$ which does not depend on the
truncation / discretization parameters such that for all $n \in
\mathbb N$ and all $K> \bar \mu
n \Delta t$, the approximating Markov chain $\widehat Z$ satisfies
$$
\mathbb P[\max_{0\leq i \leq n} |Z^{z,t_0}_{t_i}| \geq K] \leq \frac{z^2+C
n \Delta t}{ (K-\bar \mu n \Delta t)^2}. 
$$
\end{lemma}
\begin{proof}
By construction, the approximating chain $\widehat Z$ satisfies
\begin{align*}
\mathbb E[\widehat Z_{t_{i+1}}| \widehat Z_{t_i}] &=
\mu(t_i,Z_{t_i})\Delta t\\ \text{Var}\,[\widehat
Z_{t_{i+1}}| \widehat Z_{t_i}] &= D(t_i,\widehat Z_{t_i})\Delta t -
\mu^2(t_i,Z_{t_i})\Delta t^2 + \Delta t \sum_{l: y_l(t_i,\widehat
  Z_{t_i}) \in \Omega_1 \cup \Omega_2}
\omega_l(t_i,\widehat Z_{t_i}) \gamma^2(t_i,\widehat Z_{t_i},
y_l(t_i,\widehat Z_{t_i})).
\end{align*}
Therefore, the process $M_{t_i} = \widehat Z^{z,t_0}_{t_i} -
\sum_{0\leq 
  j < i} \mu(t_j, \widehat Z^{z,t_0}_{t_j}) \Delta t$
is a discrete-time martingale. By Doob's martingale inequality we then
get
\begin{align*}
\mathbb P[\max_{0\leq i \leq n} |Z^{z,t_0}_{t_i}| \geq K] &\leq \mathbb
P[\max_{0\leq i \leq n} |M_{t_i}| \geq K- \bar \mu \Delta t n]\leq
\frac{\mathbb E[M^2_{t_n}]}{(K- \bar \mu \Delta t n)^2}\\
&\leq \frac{z^2 +\mathbb E[\sum_{0<j\leq n} \text{Var}\, [\widehat
  Z^{z,t_0}_{j+1}|\widehat Z^{z,t_0}_j]]}{(K- \bar \mu \Delta t n)^2}
\\
&\leq \frac{z^2 +\mathbb E[\sum_{0<j\leq n} (D(t_j,\widehat Z_{t_j}) +
  \sum_{l:y_l\in \Omega_1 \cup \Omega_2}
  \gamma^2(t_j,\widehat Z_{t_j},y_l) \omega_l(t_j,\widehat
  Z_{t_j}))]\Delta t}{(K- \bar \mu \Delta t n)^2}.
\end{align*}
We finish the proof by applying the mean value theorem and using the
boundedness of the derivatives of $\gamma$. 
\end{proof}

\section{Proof of Proposition \ref{lem_transf}}\label{proof:trasf}
\begin{proof}
Before we start, remark that the function $A\mapsto F_{d,T}(A)$ is strictly increasing, so invertible, and infinitely differentiable: in particular 
\begin{align*}
\Phi'( A) =&\frac{\int_T^{T+d}\psi(0,s)l(s) e^{l(s) A}ds}{\int_T^{T+d}\psi(0,s) e^{l(s) A}ds}   \\
\Phi''( A) =&\frac{\left(\int_T^{T+d}\psi(0,s)l^2(s) e^{l(s) A}ds\right)\left(\int_T^{T+d}\psi(0,s) e^{l(s) A}ds\right)-\left(\int_T^{T+d}\psi(0,s)l(s) e^{l(s) A}ds\right)^2}{\left(\int_T^{T+d}\psi(0,s) e^{l(s) A}ds\right)^2}  
\end{align*}
from which we deduce
\begin{align*}
e^{-c(T+d)}\leq \Phi'(A) \leq e^{-cT} && \text{and}&& e^{-2c(T+d)} - e^{-2cT} \leq \Phi''(A) \leq e^{-2cT} - e^{-2c(T+d)} 
\end{align*}
From \itos formula, we obtain
\begin{align*}
 dZ_t = &\left(\Phi'(A_t)e^{c t}\zeta + \int_{\R}\left( \Phi(A_{t-}+e^{c t}y)-\Phi(A_{t-}) -ye^{c t}\Phi'(A_{t-})\right)\nu(dy)\right)dt \\
+&  \int_\R \left( \Phi(A_{t-}+e^{c t}y)-\Phi(A_{t-})\right) \tilde J(dydt) 
\end{align*}
or equivalently 
$$
dZ_t=\mu(t,Z_t)dt + \int\gamma(t,Z_{t-},y)\tilde J(dydt)
$$
We can now prove that $\mu$ and $\gamma$ verify the Assumptions \ref{assumptions_1}. We detail the computations only for the function $\gamma$, since similar computations can be done for $\mu$. First we remark that $z\mapsto \gamma(t,z,y)$ is differentiable and  we can compute this derivative to obtain
\begin{align*}
\partial_z\gamma(t,z,y)  =& -1 + (\Phi'(\Phi^{-1}(z)))^{-1} \Phi'(\Phi^{-1}(z)+ye^{c t}) \\
=& e^{c t}y  (\Phi'(\Phi^{-1}(z)))^{-1} \int_0^1 \Phi'' (\Phi^{-1}(z)+re^{c t})dr
\end{align*}
so that 
\begin{align*}
\left| \partial_z\gamma(t,z,y)\right|  =&\left| e^{c t}y  (\Phi'(\Phi^{-1}(z)))^{-1} \int_0^1 \Phi'' (\Phi^{-1}(z)+re^{c t})dr\right|\\
\leq & |y| e^{c T} (\inf_{A}\left| \Phi (A)\right| )^{-1}\left\|\Phi'' \right\|_\infty \leq e^{c T} e^{-c(T+d)}\left\|\Phi'' \right\|_\infty |y|\\
\leq & |y| e^{c T} e^{c(T+d)} \left(e^{-2cT} - e^{-2c(T+d)} \right) \leq e^{c d} |y|
\end{align*}
From the bounds on the first and second derivative of $\Phi$ we obtain $\sup_{t,z}|\partial_z\gamma(t,z,y)|\leq e^{c d} |y|$, which gives us the function $\rho$ introduced in Assumptions \ref{assumptions_1}. 
Again by the definition of $\Phi$ in \eqref{future} we have 
$$
\exp(e^{-c(T+d)}y)-1\leq e^{\gamma(t,z,y)}-1 \leq  e^y-1 
$$
if $y>0$ and the inverse inequality stands in force if $y<0$ which yield $\sup_{t,z}|e^{\gamma(t,z,y)}-1|\leq |e^y-1|$. According to the definition of the function $\tau$ given in Assumptions \ref{assumptions_1} and the estimations above we deduce that 
$$
\tau\left(y\right):=\max\left(\sup_{t,z}\left( |\gamma\left(t,z,y\right)|,\left|e^{\gamma\left(t,z,y\right)}-1\right|\right),\, \rho(y)\right) =e^{c d}\max\left(|y|,\left|e^y-1\right|\right)
$$
If follows then that Assumptions \ref{assumptions_1}-$\mathbf{[C,I,L]}$ hold true. For Assumption \ref{assumptions_1}-$\mathbf{[ND]}$ we have, from the definition of $\gamma$
$$
\left(e^{\gamma(t,z,y)}-1\right)^2 \geq \left( \exp(e^{-c(T+d)}y)-1\right)^2
$$
so then, for some positive $M>0$ we have
\begin{align*}
\Gamma(y):=&\int_\R\inf_{t,z}\left(e^{\gamma(t,z,y)}-1\right)^2 \nu(dy)\geq \int_\R\inf_{t,z}\left( \exp(e^{-c(T+d)}y)-1\right)^2\nu(dy)\\ 
\geq& M \int_{|y|\leq \epsilon} |y|^{1-\alpha} g(y) dy >0
\end{align*}
since $g(0^+)$ and $g(0^-)$ are strictly positive, we can select $\epsilon$ small enough and obtain
$$
\Gamma(y) \geq M \int_{|y|\leq \epsilon} |y|^{1-\alpha} dy >0
$$
\bigskip
We can derive $\gamma$ w.r.t $y$ to obtain 
\begin{align*}
\gamma_y(t,z,y) =  &e^{c t}\Phi'(\Phi^{-1}(z)+e^{c t}y)\\
\gamma_{yy}(t,z,y)=&e^{2c t}\Phi''(\Phi^{-1}(z)+e^{c t}y)
\end{align*}
so then $e^{-c(T+d)}\leq |\gamma_y(t,z,y)|\leq e^{c T}$ and $|\gamma_{yy}(t,z,y)|\leq e^{2c T}$, which proves that Assumptions \ref{assumptions_1}-$\mathbf{[RG_i]}$ holds true. For Assumption \ref{assumptions_1}-$\mathbf{[RG_{iii}]}$, one can differentiate $\gamma_y$ w.r.t. $z$ and give for it an upper bound to prove that indeed $z\to \gamma_y(t,z,y)$ is \lip continuous uniformly in $t,y$. The Assumption \ref{assumptions_1}-$\mathbf{[RG_{ii}]}$ does not hold true since trivially $\gamma_y(t,z,y) =  e^{c t}\Phi'(\Phi^{-1}(z)+e^{c t}y)\neq 1$ .\\
\noindent The last thing we need to prove is the condition 4 assumed
in Theorem \ref{maina.thm}, i.e. that $\gamma$ is 3 times differentiable w.r.t. $y$ with bounded derivatives. From the definition of $\gamma$, this is equivalent to prove that $\Phi$ is 3 times differentiable with bounded derivatives. Let us introduce
$$
p_i (A) : = \int_T^{T+d}\psi(0,s)l^i(s) e^{l(s) A}ds 
$$
so that $ \Phi'(A) = p_1(A) / p_0(A)$ and $\Phi''(A) = (p_2(A)p_0(A) - p_1(A) ^2)/ p_0(A)^2$. By remarking that $p_i'(A) = p_{i+1}(A)$, we can differentiate $\Phi''$ to obtain
$$
\Phi^{(iii)}(A) = \frac{p_3(A)p_0(A)+ p_2(A)p_1(A) - 2p_1(A)p_2(A)}{p_0(A)^2} - 2 \Phi''(A) \Phi'(A)
$$
with $\Phi'$ and $\Phi''$ bounded as already proved. With the same type of computation it is straightforward to prove that $\Phi^{(iii)}(A)$ is also bounded.

\end{proof}

\end{document}